%% file: main.tex
\documentclass[10pt]{article}

\usepackage[utf8]{inputenc}
\usepackage[T1]{fontenc}
\usepackage{lmodern}
\usepackage[left=1in,right=1in,top=1in,bottom=1in]{geometry}

\usepackage{microtype}

\usepackage{amsmath}
\usepackage{amsthm}

\usepackage{amssymb}
\usepackage{amsthm}
\usepackage{thmtools}
\usepackage{thm-restate}

\usepackage{titlesec}
\usepackage{fancybox}
\usepackage{xcolor}
\usepackage{xspace}
\usepackage{enumerate}
\usepackage{graphicx}
\usepackage{epstopdf}

\usepackage[ocgcolorlinks]{hyperref} % Option ocgcolorlinks makes links only colored on
\usepackage{cleveref}
\usepackage{authblk}
\usepackage{algorithm}
\usepackage{hyperref}
\usepackage{paralist}

\titlespacing{\paragraph}{%
	0pt}{%              left margin
	0.5\baselineskip}{% space before (vertical)
	1em}%               space after (horizontal)

\makeatletter

\def\moverlay{\mathpalette\mov@rlay}
\def\mov@rlay#1#2{\leavevmode\vtop{%
		\baselineskip\z@skip \lineskiplimit-\maxdimen
		\ialign{\hfil$\m@th#1##$\hfil\cr#2\crcr}}}
\newcommand{\charfusion}[3][\mathord]{
	#1{\ifx#1\mathop\vphantom{#2}\fi
		\mathpalette\mov@rlay{#2\cr#3}
	}
	\ifx#1\mathop\expandafter\displaylimits\fi}
\makeatother

\colorlet{DarkRed}{red!60!black}
\colorlet{DarkGreen}{green!50!black}
\colorlet{DarkBlue}{blue!60!black}

\hypersetup{
	linkcolor = DarkRed,
	citecolor = DarkGreen,
	urlcolor = DarkBlue,
	bookmarksnumbered = true,
	linktocpage = true
}

\declaretheorem[numberwithin=section]{theorem}
\declaretheorem[numberlike=theorem]{lemma}

\declaretheorem[numberlike=theorem]{corollary}

\input macros

\newcommand{\rot}{\mathrm{T}}

\newcommand{\source}{s_0}
\newcommand{\sink}{s_1}
\newcommand{\lref}[1]{(\ref{#1})}
\renewcommand{\epsilon}{\varepsilon}
\newcommand{\Y}{Y}

\newcommand{\Xstar}{X_{\star}}

\newcommand{\PExt}{P^{(\mathrm{ext})}}

\newcommand{\PSI}[1]{\Psi^{(#1)}}
\newcommand{\UPSILON}[1]{\Upsilon^{(#1)}}

\newcommand{\p}[1]{p^{(#1)}}
\newcommand{\q}[1]{q^{(#1)}}

\newcommand{\x}[1]{x^{(#1)}}
\newcommand{\xl}{\x{\ell}}
\newcommand{\xlplus}{\x{\ell + 1}}

\newcommand{\y}[1]{y^{(#1)}}
\newcommand{\z}[1]{z^{(#1)}}

\newcommand{\pl}{\p{\ell}}
\newcommand{\ql}{\q{\ell}}
\newcommand{\Ll}{L^{(\ell)}}

\newcommand{\Rl}{R^{(\ell)}}
\newcommand{\Xl}{X^{(\ell)}}

\newcommand{\xprm}{x^{\prime}}

\newcommand{\h}[1]{h^{(#1)}}
\newcommand{\hl}{\h{\ell}}

\newcommand{\hprm}{h_{0}}

\newcommand{\ah}[1]{\alpha^{(#1)}}
\newcommand{\ahl}{\ah{\ell}}
\newcommand{\ahlplus}{\ah{\ell + 1}}
\newcommand{\ahz}{\alpha^{(0)}}

\newcommand{\vZeros}{\mathbf{0}}

\newcommand{\rhoA}{\rho_{A}}

\newcommand{\rk}{r^{(k)}} 
\newcommand{\rkp}{r^{(k+1)}}
\newcommand{\xkp}{x^{(k+1)}}
\newcommand{\qk}{q^{(k)}} 

\newcommand{\ex}{\left( \begin{array}{c} x \\ x' \end{array}\right)}

\newcommand{\xkm}{x^{(k-1)}} 
\newcommand{\skm}{s^{(k-1)}}
\newcommand{\sk}{s^{(k)}}
\newcommand{\xk}{x^{(k)}}
\newcommand{\xz}{x^{(0)}}
\newcommand{\fk}{f^{(k)}}
\newcommand{\fstr}{f^{\star}}
\newcommand{\epsStr}{\epsilon^{\star}}

\begin{document}

\title{Two Results on Slime Mold Computations}

\author[1]{Ruben Becker}
\author[2]{Vincenzo Bonifaci}
\author[1]{Andreas Karrenbauer}
\author[1]{\\Pavel Kolev\thanks{This work has been funded by the Cluster of Excellence ``Multimodal Computing and Interaction'' within the Excellence Initiative of the German Federal Government.}}
\author[1]{Kurt Mehlhorn}

\affil[1]{Max Planck Institute for Informatics, Saarland Informatics Campus, Saarbr{\"u}cken, Germany \{ruben,karrenba,pkolev,mehlhorn\}@mpi-inf.mpg.de}
\affil[2]{Institute for the Analysis of Systems and Informatics,
	National Research Council of Italy (IASI-CNR), Rome, Italy,
	vincenzo.bonifaci@iasi.cnr.it}

\date{}
\maketitle

\begin{abstract}
	We present two results on slime mold computations. In wet-lab experiments (Nature'00) by Nakagaki et al. the slime mold Physarum polycephalum demonstrated its ability to solve shortest path problems. Biologists proposed a mathematical model, a system of differential equations, for the slime's adaption process (J. Theoretical Biology'07). It was shown that the process convergences to the shortest path (J. Theoretical Biology'12) for all graphs. We show that the dynamics actually converges for a much wider class of problems, namely undirected linear programs with a non-negative cost vector. 

	Combinatorial optimization researchers took the dynamics describing slime behavior as an inspiration	for an optimization method and showed that its discretization can $\varepsilon$-approximately solve linear programs with positive cost vector (ITCS'16). Their analysis requires a feasible starting point, a step size depending linearly on $\varepsilon$, and a number of steps with quartic dependence on $\mathrm{opt}/(\varepsilon\Phi)$, where $\Phi$ is the difference between the smallest cost of a non-optimal basic feasible solution and the optimal cost ($\mathrm{opt}$).

	We give a refined analysis showing that the dynamics initialized with any strongly dominating point converges to the set of optimal solutions. Moreover, we strengthen the convergence rate bounds and prove that the step size is independent of $\varepsilon$, and the number of steps depends logarithmically on $1/\varepsilon$ and quadratically on $\mathrm{opt}/\Phi$.
\end{abstract}

\newpage
\tableofcontents
\newpage

\input{KMintroduction}

\newpage

\input{SimpleInstances}
\newpage

\input{undirected}

\newpage

\input{directed}
\input{Preconditioning}

\input{LowerBound}

\section{Conclusions and Open Problems}

We proved convergence of the Physarum dynamics under conditions \eqref{c is nonnegative} to
\eqref{positive start vector}. One of the reviewers asked us to discuss whether these conditions
can be further relaxed.
1) Can we allow a start vector with components equal to zero? 
This is not interesting because $x_e(0) = 0$ implies $x_e(t) = 0$ for all $t$, 
i.e., the column $e$ is simply removed from the system. 
2) Can we allow that there is an element $f$ in the kernel of $A$ with $c_e f_e = 0$ for every $e$? 
Then the minimum energy solution is no longer unique; cf. Lemma~\ref{q is kernel-free}. 
3) Can we allow components of $c$ to be negative? We did not explore this direction because we were thinking of $r_e$ as resistances and negative resistances do not make sense. 
However, the mathematics might go through as long as non-zero vectors in the kernel of $A$ 
have positive cost.

An interesting open question for the Physarum-inspired dynamics~\eqref{eq:DdirD} 
is whether it converges to the set of optimal solutions under conditions
\eqref{c is nonnegative} to \eqref{positive start vector}.

\bibliographystyle{alpha}
\bibliography{references}

\end{document}

%% file: macros.tex
\renewcommand{\epsilon}{\varepsilon}

\newcommand{\myurl}[1]{{\footnotesize \url{#1}}}

\newcommand{\donotshow}[1]{}

\newcommand{\ignore}[1]{}

\newcommand{\assign}{\mathbin{\raisebox{0.05ex}{\mbox{\rm :}}\!\!=}}

\newcommand{\argmin}{\mathop {\rm argmin}}
\providecommand{\R}{\mathbb{R}}
\newcommand{\N}{\mathbb{N}}
\newcommand{\Z}{\mathbb{Z}}

\newcommand {\abs}[1] {| #1 |}

\newcommand{\XX}{\mathcal{X}}
\providecommand{\Copt}{C_{\star}}

\providecommand{\sset}[1]{\{\, #1 \,\}}
\providecommand{\sset}[1]{\{\, #1 \,\}}
\providecommand{\smSet}[2]{ \{\, #1 \mbox{ {\rm :} } #2 \,\}  }
\newcommand{\set}[2]{ \left\{\, #1 \mbox{ {\rm :} } #2 \,\right\}  }
\providecommand{\set}[2]{ \left\{\, #1 \,\, \mbox{{\rm :}}\,\, #2 \, \right\}  }

\newcommand{\norm}[1]{\vert\!\vert #1 \vert\!\vert}
\newcommand{\onenorm}[1]{{\norm{#1}_{1}}}
\newcommand{\inftynorm}[1]{\norm{#1}_{\infty}}

\newcommand{\bfzero}{{\bf 0}}
\newcommand{\diag}{\mathrm{diag}}
\newcommand{\qone}{q^{(1)}}  \newcommand{\qtwo}{q^{(2)}}

\newcommand{\dist}{\mathrm{dist}}

\newcommand{\cmax}{c_{\max}}
\newcommand{\cmin}{c_{\min}}
\newcommand{\D}{D}
\newcommand{\infnorm}[1]{\norm{#1}_\infty}

\newcommand{\opt}{{\mathrm{opt}}}
\newcommand{\OPT}{\mathcal{S}_{\opt}}
\newcommand{\Eopt}{E_{\opt}}

\newcommand{\X}{X}
\newcommand{\cost}{{\mathrm{cost}}}
\newcommand{\E}{E}

\DeclareMathOperator{\rank}{{\mathrm{rank}}}
\DeclareMathOperator{\supp}{{\mathrm{supp}}}

\makeatletter
\renewcommand{\paragraph}{%
 \@startsection{paragraph}{4}%
 {\z@}{1.2ex \@plus .5ex \@minus .1ex}{-.7em}%
 {\normalfont\normalsize\bfseries}%
}
\makeatother

\newcommand{\Xdom}{\X_{\mathrm{dom}}}

%% file: KMintroduction.tex
%!TEX root=./main.tex

\section{Introduction}

We present two results on slime mold computations, one on the biologically-grounded model and one on the biologically-inspired model. The first model was introduced by biologists to capture the slime's apparent ability to compute shortest paths. We show that the dynamics can actually do more. It can solve a wide class of linear programs with nonnegative cost vectors. The latter model was designed as an optimization technique inspired by the former model. We present an improved convergence result for its discretization. The two models are introduced and our results are stated in Sections~\ref{biologically-grounded} and~\ref{biologically-inspired} respectively. The results on the former model are shown in Sections~\ref{simple instances} and~\ref{general instances}, the results on the latter model are shown in Section~\ref{discretization}. 

\subsection{The Biologically-Grounded Model	}\label{biologically-grounded}

\begin{figure}[t]
\begin{center}
\includegraphics[width=0.42\textwidth]{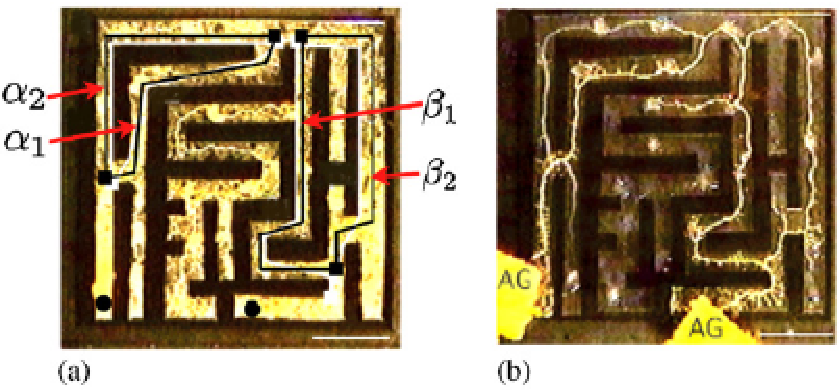}
\hspace{0.3cm}
\includegraphics[width=0.40\textwidth]{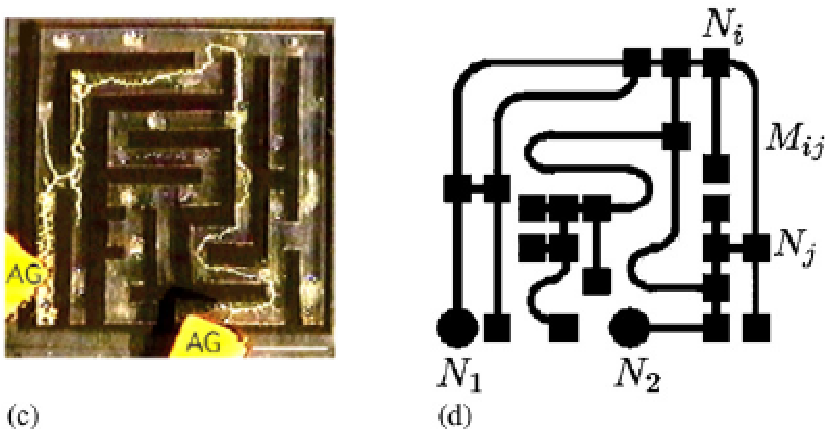}
\end{center}
\caption{\label{fig:maze} The experiment in~\cite{Nakagaki-Yamada-Toth}
(reprinted from there): (a) shows the maze uniformly covered by Physarum;
yellow color indicates presence of Physarum. Food (oatmeal) is provided at the
locations labeled AG. After a while the mold retracts to the shortest path
connecting the food sources as shown in (b) and (c). (d) shows the underlying
abstract graph. The video~\cite{Physarum-Video} shows the
experiment.
}
\end{figure}

\emph{Physarum polycephalum} is a slime mold that apparently is able to solve
shortest path problems. Nakagaki, Yamada, and T\'{o}th~\cite{Nakagaki-Yamada-Toth} report about the following experiment; see Figure~\ref{fig:maze}. They built a maze, covered it by pieces of Physarum (the slime can be cut into pieces which will reunite if brought into vicinity), and then fed the slime with oatmeal at two locations. After a few hours the slime retracted to a path that follows the shortest path in the maze connecting the food sources. The authors report that they repeated the experiment with different mazes; in all experiments, Physarum retracted to the shortest path.

The paper~\cite{Tero-Kobayashi-Nakagaki} proposes a mathematical model for the behavior of the slime and argues extensively that the model is adequate. Physarum is modeled as an electrical network with time varying resistors. We have a simple \emph{undirected} graph $G = (N,E)$ with distinguished nodes $\source$ and $\sink$ modeling the food sources. Each edge $e \in E$ has a positive length $c_e$ and a positive capacity $x_e(t)$; $c_e$ is fixed, but $x_e(t)$ is a function of time. The resistance $r_e(t)$ of $e$ is $r_e(t) = c_e/x_e(t)$. In the electrical network defined by these resistances, a current of value 1 is forced from $\source$ to $\sink$. For an (arbitrarily oriented) edge $e = (u,v)$, let $q_e(t)$ be the resulting current over $e$. Then, the capacity of $e$ evolves according to the differential equation
\begin{equation}\label{PD}
\dot{x}_e(t) = | q_e(t) | - x_e(t),
\end{equation}
where $\dot{x}_e$ is the derivative of $x_e$ with respect to time. In equilibrium ($\dot{x}_e = 0$ for all $e$), the flow through any edge is equal to its capacity. In non-equilibrium, the capacity grows (shrinks) if the absolute value of the flow is larger (smaller) than the capacity. In the sequel, we will mostly drop the argument $t$ as is customary in the treatment of dynamical systems. We will also write $q$ for the vector with components $q_e$. It is well-known that the electrical flow $q$ is the feasible flow minimizing energy dissipation $\sum_e r_e q_e^2$ (Thomson's principle). 

We refer to the dynamics above as \emph{biologically-grounded}, as it was introduced by biologists to model the behavior of a biological system. Miyaji and Ohnishi were the first to analyze convergence for special graphs (parallel links and planar graphs with source and sink on the same face) in~\cite{Miyaji-Ohnishi}. In~\cite{Physarum} convergence was proven for \emph{all} graphs. We state the result from~\cite{Physarum} for the special case that the shortest path is unique.

\begin{theorem}[\cite{Physarum}] 
	Assume $x_e(0) > 0$ and $c_e > 0$ for all $e$, 
	and that the undirected shortest path $P^*$ from $\source$ to $\sink$ w.r.t.~the 
	cost vector $c$ is unique. Then $x(t)$ in~\lref{PD} converges to $P^*$. 
	Namely, $x_e(t) \rightarrow 1$ for $e \in P^*$ and $x_e \rightarrow 0$ 
	for $e \not\in P^*$ as $t \rightarrow \infty$.
\end{theorem}

\cite{Physarum} also proves an analogous result for the undirected transportation problem;~\cite{Bonifaci-Physarum} simplified the argument under additional assumptions. The paper~\cite{Bonifaci-RefinedModel} studies a more general dynamics and proves convergence for parallel links.

In this paper, we extend this result to \emph{non-negative undirected linear programs}
\begin{equation}\label{ULP}
    \min\{c^{\rot} x : Af = b,\ \abs{f} \leq x\},
\end{equation}
where $A \in \R^{n \times m}$, $b \in \R^n$, $x \in \R^m$, $c \in \R_{\ge 0}^m$, and
the absolute values are taken componentwise. 
Undirected LPs can model a wide range of problems, e.g., optimization problems
such as shortest path and min-cost flow in undirected graphs, and the Basis Pursuit problem
in signal processing~\cite{CDS98}.

We use $n$ for the number of rows of $A$ and $m$ 
for the number of columns, since this notation is appropriate when $A$ is the node-edge-incidence 
matrix of a graph. 
A vector $f\in\R^{m}$ is \emph{feasible} if $A f = b$.
We assume that the system $Af = b$ has a feasible solution and every nonzero vector $f$ 
in the kernel~\footnote{ 
	The kernel of a matrix $A$ consists of all solutions to the system $Ax = 0$.} 
of $A$ has positive cost $\sum_e c_e \abs{f_e} > 0$.
The vector $q\in\R^{m}$ in~\eqref{PD} is now the \emph{minimum energy feasible solution}
\begin{equation}\label{Defq}
q(t) = \argmin_{f \in \R^m} \left\{\sum_{e: x_e \neq 0} \frac{c_e}{x_e(t)} f_e^2:Af = b \wedge f_e = 0 \text{ whenever } x_e = 0\right\}.
\end{equation}
We remark that $q$ is unique; see Subsection~\ref{minimum energy solution}. 
If $A$ is the incidence matrix of a graph (the column corresponding to an edge $e$ has one entry $+1$, one entry $-1$ and all other entries are equal to zero), \eqref{ULP} is a transshipment problem with flow sources and sinks encoded by a demand vector $b$. The condition that there is no solution in the kernel
of $A$ with $c_ef_e = 0$ for all $e$ states that every cycle contains at least one edge of positive cost. In that setting, $q(t)$ as defined by \eqref{Defq} coincides with the electrical flow induced by resistors of value $c_e/x_e(t)$. We can now state our first main result.

\begin{theorem}\label{intro_thm}
Let $c \ge 0$ satisfy $c^{\rot} \abs{f} > 0$ for every non-zero $f$ in the kernel of $A$. Let $x^*$ be an optimum solution of~\eqref{ULP} and let $\Xstar$ be the set of optimum solutions. Assume $x(0) > 0$. The following holds for the dynamics~\lref{PD} with $q$ as in~\lref{Defq}:
\begin{compactenum}[\mbox{}\hspace{\parindent}(i)]
\item The solution $x(t)$ exists for all $t \ge 0$. 
\item The cost $c^{\rot} x(t)$ converges to $c^{\rot} x^*$ as $t$ goes to infinity.
\item The vector $x(t)$ converges to $\Xstar$, i.e., 
$\lim_{t \rightarrow \infty} \inf \set{\norm{x(t) - x'}}{x' \in \Xstar} \rightarrow 0$.
\item For all $e$ with $c_e>0$, $x_e(t) - \abs{q_e(t)}$ converges to zero as $t$ goes to infinity.~\footnote{
	We conjecture that this also holds for the indices $e$ with $c_e = 0$.} If $x^*$ is unique, $x(t)$ and $q(t)$ converge to $x^*$ as $t$ goes to infinity.
\end{compactenum}
\end{theorem}
Item (i) was previously shown in~\cite{SV-IRLS} for the case of a strictly positive cost vector. 
The result in~\cite{SV-IRLS} is actually stated only for the all-ones cost vector $c = \bf{1}$.
The case of a general positive cost vector reduces to this special case by rescaling the solution vector $x$. Item $(i)$ for the more general cost vector and items $(ii)$ to $(iv)$ are new. We stress that the dynamics~\lref{PD} is biologically-grounded. It was proposed to model a biological system and not as an optimization method. Nevertheless, it can solve a large class of non-negative LPs. Table~\ref{undirected results} summarizes our first main result and puts it into context. 

\begin{table}[t]
	\centering
	\begin{tabular}{|c|c|c|c|c|}
		\hline
		\textbf{Reference} & \begin{tabular}[c]{@{}c@{}}\textbf{Problem}\\ (Undirected Case) \end{tabular}       & \begin{tabular}[c]{@{}c@{}}\textbf{Existence}\\ \textbf{of Solution}\end{tabular} & \begin{tabular}[c]{@{}c@{}}\textbf{Convergence}\\ \textbf{ to} $\mathrm{OPT}$\end{tabular} & \textbf{Comments}                                                                                                            \\ \hline
		\cite{Miyaji-Ohnishi}      & {Shortest Path} & Yes                   & Yes                           & parallel edges, planar graphs\\ \hline
		\cite{Physarum}     & {Shortest Path} & Yes              & Yes                     & all graphs                                                                                                        \\ \hline
		\cite{SV-IRLS}     & {Positive LP} & Yes           &  No                    & $c>0$                                                                                                               \\ \hline
		\begin{tabular}[c]{@{}c@{}}\textbf{Our}\\ \textbf{Result}\end{tabular} & {Nonnegative LP}   & Yes                   & {Yes}                           & \begin{tabular}[c]{@{}c@{}}{1) $c\geq0$}\\ {2) $\forall v\in\mathrm{ker}(A)\,:\, c^{\rot}|v|>0$}\end{tabular}\\ \hline
	\end{tabular}
	\smallskip
	\caption{Convergence results for the continuous undirected Physarum dynamics~\lref{PD}.}
	\label{undirected results}
\end{table}

Sections~\ref{Simple Instances} and~\ref{General Case} are devoted to the proof of our first main theorem. For ease of exposition, we present the proof in two steps. In Section~\ref{Simple Instances}, we give a proof under the following simplifying assumptions. 
\begin{compactenum}[\hspace{\parindent}(A)]
	\item $c > 0$,\label{positive c}
	\item The basic feasible solutions of~\eqref{ULP} have distinct cost,\label{distinctCost}
	\item We start with a positive vector $x(0) \in \Xdom \assign \set{x \in \R^n}{\text{there is a feasible $f$ with $\abs{f} \le x$}}$\label{x0 is dominating}. 
\end{compactenum}
Section~\ref{Simple Instances} generalizes~\cite{Bonifaci-Physarum}. For the undirected shortest path problem, condition \eqref{distinctCost} states that all simple undirected source-sink paths have distinct cost and condition \eqref{x0 is dominating} states that all source-sink cuts have a capacity of at least one at time zero (and hence at all times). The existence of a solution with domain $[0,\infty)$ was already shown in~\cite{SV-IRLS}. 
We will show that $\Xdom$ is an invariant set, i.e., the solution stays in $\Xdom$ for all times, and that $E(x) = \sum_e r_e x_e^2 = \sum_e c_e x_e$ is 
a Lyapunov function~\footnote{
	Lyapunov functions are a main tool for proving convergence of dynamical systems. It is a function $L(t)$ mapping the state $x(t)$ of the system to a non-negative real such that $\dot{L} \le 0$ and $\dot{L} = 0$ iff $\dot{x} = 0$. It is an ``art'' to find a Lyapunov function for a concrete dynamical system.}~\cite{LaSalle,Teschl12} for the dynamics~\lref{PD}, i.e., $\dot{E} \le 0$ and $\dot{E} = 0$ if and only if $\dot{x} = 0$. It follows from general theorems about dynamical systems that the dynamics converges to a fixed point of~\lref{PD}. The fixed points are precisely the vectors $\abs{f}$ where $f$ is a feasible solution of~\lref{ULP}. A final argument establishes that the dynamics converges to a fixed point of minimum cost.

In Section~\ref{General Case}, we prove the general case of the first main theorem. We assume 
\begin{compactenum}[\hspace{\parindent}(A)]\setcounter{enumi}{3}
	\item $c \ge 0$, \label{c is nonnegative}
	\item $\cost(z)=c^{\rot}|z| > 0$ for every non-zero vector $z$ in the kernel of $A$, \label{elements in kernel have positive cost}
	\item We start with a positive vector $x(0) > 0$.\label{positive start vector}
\end{compactenum}
Section~\ref{General Case} generalizes~\cite{Physarum} in two directions. First, we treat general undirected LPs and not just the undirected shortest path problem, respectively, the transshipment problem. Second, we replace the condition $c > 0$ by the requirement $c \ge 0$ and every non-zero vector in the kernel of $A$ has positive cost. For the undirected shortest path problem, the latter condition states that the underlying undirected graph has no zero-cost cycle. Section~\ref{General Case} is technically considerably more difficult than Section~\ref{Simple Instances}. We first establish the existence of a solution with domain $[0,\infty)$. To this end, we derive a closed formula for the minimum energy feasible solution and prove that the mapping $x \mapsto q$ is locally Lipschitz. Existence of a solution with domain $[0,\infty)$ follows by standard arguments. We then show that $\Xdom$ is an attractor, i.e., the solution $x(t)$ converges to $\Xdom$. We next characterize equilibrium points and exhibit a Lyapunov function. The Lyapunov function is a normalized version of $E(x)$. The normalization factor is equal to the optimal value of the linear program $\max \set{\alpha}{Af = \alpha b,\ \abs{f} \le x}$ in the variables $f$ and $\alpha$. Convergence to an equilibrium point follows from the existence of a Lyapunov function.  A final argument establishes that the dynamics converges to a fixed point of minimum cost.

\subsection{The Biologically-Inspired Model}\label{biologically-inspired}

Ito et al.~\cite{Ito-Convergence-Physarum} initiated the study of the dynamics
\begin{equation}\label{directed dynamics}  
	\dot{x}(t) = q(t) - x(t). 
\end{equation}
We refer to this dynamics as the directed dynamics in contrast to the undirected dynamics~\lref{PD}.
The directed dynamics is \emph{biologically-inspired} -- the similarity to~\lref{PD} is the inspiration. It was never claimed to model the behavior of a biological system. Rather, it was introduced as a biologically-inspired optimization method. The work in~\cite{Ito-Convergence-Physarum} shows convergence of this directed dynamics~\lref{directed dynamics} for the directed shortest path problem and~\cite{Johannson-Zou,SV-LP,Bonifaci-LP} show convergence for general \emph{positive linear programs}, i.e., linear programs with positive cost vector $c > 0$ of the form
\begin{equation}\label{OLP}
    \min\{c^{\rot} x : Ax = b,\ x \ge 0\}.
\end{equation}

The \emph{discrete versions} of both dynamics define sequences $\x{t}$, $t = 0,1,2,\ldots$ through
\begin{align}
x^{(t+1)} & = (1 - h^{(t)})\x{t} + h^{(t)} q^{(t)}        &\text{discrete directed dynamics};\label{eq:DdirD} \\
x^{(t+1)} & = (1 - h^{(t)})\x{t} + h^{(t)} |q^{(t)}|   &\text{discrete undirected dynamics}\label{eq:DUndirD},
\end{align}
where $h^{(t)}$ is the step size and $q^{(t)}$ is the minimum energy feasible solution as in~\eqref{Defq}. 
For the discrete dynamics, we can ask complexity questions. This is particularly relevant for 
the discrete directed dynamics as it was designed as an biologically-inspired optimization method.

For completeness, we review the state-of-the-art results for the discrete undirected dynamics. For the undirected shortest path problem, the convergence of the discrete undirected dynamics~\eqref{eq:DUndirD} was shown in~\cite{Physarum-Complexity-Bounds}. The convergence proof gives an upper bound on the step size and on the number of steps required until an $\epsilon$-approximation of the optimum is obtained. \cite{SV-Flow} extends the result to the transshipment problem and \cite{SV-IRLS} further generalizes the result to the case of positive LPs.
The paper~\cite{SV-Flow} is related to our first result.
It shows convergence of the discretized undirected dynamics~\eqref{eq:DUndirD}, we show convergence of the continuous undirected dynamics~\eqref{PD} for a more general cost vector.
\bigskip

We come to the discrete directed Physarum-inspired dynamics~\lref{eq:DdirD}. Similarly to the undirected setting, 
Becchetti et al.~\cite{Physarum-Complexity-Bounds} showed the convergence of \lref{eq:DdirD} 
for the shortest path problem.
Straszak and Vishnoi extended the analysis
to the transshipment problem~\cite{SV-Flow} and positive LPs~\cite{SV-LP}.
\begin{theorem}\cite[Theorem 1.3]{SV-LP}\label{thm:Prev}
	Let $A\in\Z^{n\times m}$ have full row rank $(n\leq m)$, $b\in\Z^{n}$,
	$c\in\Z_{>0}^{m}$, and let 
	$\D_S := \max\{|\det(M)|\,:\, M\text{ is a square sub-matrix of }A\}$.\footnote{
		Using Lemma~\ref{lem:fsDg}, the dependence on $D_S$ can be improved 
		to a scale-independent determinant $D$, defined in \eqref{eq:defD}.
		For further details, we refer the reader to Subsection~\ref{subsec:subsecUsefulLemmas}.}
	Suppose the Physarum-inspired dynamics \eqref{eq:DdirD} is initialized with
	a feasible point $\xz $ of~\lref{OLP} such that $M^{-1}\le \xz \le M$
	and $c^{\rot}\xz \le M\cdot \opt $ for some $M\ge1$, where $\opt$ denotes 
	the optimum cost of \eqref{OLP}.
	Then, for any $\epsilon>0$
	and step size $h\leq\epsilon/(\sqrt{6}\onenorm{c}\D_{S})^{2}$,
	after $k=O((\epsilon h)^{-2}\ln M)$ steps, 
	$\x{k}$ is a feasible solution with $c^{\rot}\x{k}\le(1+\epsilon)\opt $.
\end{theorem}

Theorem~\ref{thm:Prev} gives an algorithm that computes a $(1+\epsilon)$-approximation 
to the optimal cost of \eqref{OLP}. In comparison to~\cite{Physarum-Complexity-Bounds,SV-Flow}, it has
several shortcomings. First, it requires a feasible starting point. Second, the step size depends linearly
on $\epsilon$. Third, the number of steps required to reach an $\epsilon$-approximation has a quartic dependence on $\opt/(\epsilon\Phi)$. In contrast, the analysis
in \cite{Physarum-Complexity-Bounds,SV-Flow} yields a step size independent of $\epsilon$ and 
a number of steps that depends only logarithmically on $1/\epsilon$, see Table~\ref{tableDiscrete}.

We overcome these shortcomings. Before we can state our result, we need some notation.
Let $\Xstar$ be the set of optimal solutions to \eqref{OLP}. 
The distance of a capacity vector $x$ to $\Xstar$ is defined as 
\[
	\dist(x,\Xstar):=\inf\{\inftynorm{x-\xprm}\,:\,\xprm\in\Xstar\}.
\] 
Let  $\gamma_{A}:=\gcd(\{A_{ij}:A_{ij}\neq0\})\in\Z_{>0}$ and
\begin{equation}\label{eq:defD}
\D := \max\left\{ \left|\det\left(M\right)\right|\, :\, 
M\text{ is a square submatrix of }A/\gamma_{A}
\text{ with dimension }n-1\text{ or }n\right\}.
\end{equation}
Let $\mathcal{N}$ be the set of non-optimal basic feasible solution of \eqref{OLP} and
\begin{equation}\label{eq:defPHI}
\Phi:=\min_{g\in\mathcal{N}}c^{\rot}g-\opt\geq1/(\D \gamma_{A})^{2},
\end{equation}
where the inequality is well known~\cite[Lemma 8.6]{PapadimitriouSteiglitz82}. For completeness, we present a proof in Subsection~\ref{subsec:epsCloseToOPT}.
Informally, our second main result proves the following properties of the Physarum-inspired dynamics \eqref{eq:DdirD}:
\begin{compactenum}[\hspace{\parindent}(i)]
	\item For any $\epsilon>0$ and any strongly dominating starting point\footnote{
		We postpone the definition of strongly dominating capacity vector to Section~\ref{subsec:SDCV}. Every scaled feasible solution is strongly dominating. In the shortest path problem, a capacity vector $x$ is strongly dominating if every source-sink cut $(S,\overline{S})$ has positive directed capacity, i.e., $\sum_{a\in E(S,\overline{S})}x_a-\sum_{a\in E(\overline{S},S) } x_a > 0$.}
 $\xz$, 
		there is a fixed step size $h(\xz)$ such that the Physarum-inspired dynamics \eqref{eq:DdirD} 
		initialized with $\xz$ and $h(\xz)$ converges to $\Xstar$, 
		i.e., $\dist(\x{k},\Xstar)<\epsilon/(\D\gamma_{A})$ for large enough $k$.
	\item The step size can be chosen \textit{independently} of $\epsilon$.
	\item The number of steps $k$ depends \textit{logarithmically} on $1/\epsilon$ and 
	\textit{quadratically} on $\opt/\Phi$.
	\item The efficiency bounds depend on a \emph{scale-invariant} determinant\footnote{
	Note that $(\gamma_{A})^{n-1}\D\leq\D_{S}\leq(\gamma_{A})^n\D$,
	and thus $\D$ yields an exponential improvement over $\D_{S}$, whenever $\gamma_{A}\geq2$.} 
	$D$. 
\end{compactenum}
\vskip 5pt

In Section~\ref{SimpleLowerBound}, we establish a corresponding lower bound. 
We show that for the Physarum-inspired dynamics \eqref{eq:DdirD} 
to compute a point $\x{k}$ such that $\dist(\x{k},\Xstar)<\epsilon$, 
the number of steps required for computing an $\epsilon$-approximation 
has to grow linearly in $\opt/(h\Phi)$ and $\ln(1/\epsilon)$, 
i.e. $k\geq\Omega( \opt\cdot(h\Phi)^{-1}\cdot\ln(1/\epsilon))$.
Table~\ref{tableDiscrete} puts our results into context.

We state now our second main result for the special case of a feasible starting point,
and we provide the full version in Theorem~\ref{thm_main_full} which applies for arbitrary 
strongly dominating starting point, see Section~\ref{discretization}.
We use the following constants in the statement of the bounds. 
\begin{enumerate}[\mbox{}\hspace{\parindent}(i)]
	\item $h_0 := c_{\min}/(4\D \onenorm{c})$, where $c_{\min}:=\min_{i}\{c_i\}$;
	\item $\PSI{0} := \max\{m\D^{2}\|b/\gamma_{A}\|_{1},\|\x{0} \|_{\infty}\}$;
	\item $C_{1} := \D\|b/\gamma_{A}\|_{1}\onenorm{c}$,
	$\,C_2 := 8^{2}m^{2}n\D^{5}\gamma_{A}^{2}\infnorm{A}\lVert b\rVert_{1}\,$
	and $\,C_3:=\D^3\gamma_{A}\onenorm{b}\onenorm{c}$.
\end{enumerate}

\begin{theorem}\label{thm_main}
	Suppose $A\in\Z^{n\times m}$ has full row rank $(n\leq m)$, $b\in\Z^{n}$, 
	$c\in\Z_{>0}^{m}$ and $\epsilon\in(0,1)$. Given a feasible starting point $\xz>0$ 
	the Physarum-inspired dynamics \eqref{eq:DdirD} with step size $h \leq (\Phi/\opt)\cdot h_{0}^{2}/2$ 
	outputs for any
	$k\geq4C_{1}/(h\Phi)\cdot\ln(C_{2}\PSI{0}/(\epsilon\cdot\min\{1,x_{\min}^{(0)}\}))$
	a feasible $x^{(k)}>0$ such that 
	$\dist(x^{(k)},\Xstar)<\epsilon/(\D \gamma_{A})$. 
\end{theorem}

\begin{table}[t]
	\centering
	\begin{tabular}{|c|c|c|c|c|}
		\hline
		\textbf{Reference} & \begin{tabular}[c]{@{}c@{}}\textbf{Problem}\\ (Directed Case) \end{tabular}
		& \textbf{step size} $h$ &  \textbf{number of steps} $k$ & \textbf{Guarantee}  \\ \hline\hline
		\cite{Physarum-Complexity-Bounds}    & Shortest Path  & indep. of $\epsilon$   & \begin{tabular}[c]{@{}c@{}}$\mathrm{poly}(m,n,\onenorm{c},\onenorm{\xz })$\\ $\cdot \ln(1/\epsilon)$\end{tabular} & $\dist(\x{k},\Xstar)<\epsilon$ \\ \hline
		\cite{SV-Flow}     & Transshipment                                                               & indep. of $\epsilon$             & \begin{tabular}[c]{@{}c@{}}$\mathrm{poly}(m,n,\onenorm{c},\onenorm{b},\onenorm{\xz })$\\ $\cdot \ln(1/\epsilon)$\end{tabular}                                          & $\dist(\x{k},\Xstar)<\epsilon$              \\ \hline \hline
		\cite{SV-LP}     & Positive LP        & depends on $\epsilon$ & \begin{tabular}[c]{@{}c@{}} $\mathrm{poly}(\onenorm{c},\D_S,\ln \onenorm{\xz })$\\ $\cdot$ $1/(\Phi\epsilon)^4$ \end{tabular}                                              & 
		\begin{tabular}[c]{@{}c@{}} 
			$c^{\rot}\x{k}\leq(1+\epsilon)\opt $\\
			$c^{\rot}\x{k}<\min_{g\in\mathcal{N}}c^{\rot}g$
		\end{tabular} \\ \hline
		\begin{tabular}[c]{@{}c@{}}\textbf{Our}\\ \textbf{Result}\end{tabular} & Positive LP        & indep. of $\epsilon$       & \begin{tabular}[c]{@{}c@{}}$\mathrm{poly}(\onenorm{c}	,\onenorm{b},\D,\gamma_{A},\ln \onenorm{\xz })$\\ $\cdot\Phi^{-2}\ln(1/\epsilon)$  \end{tabular}                                       & $\dist(\x{k},\Xstar)<\frac{\epsilon}{\D\gamma_{A}}$              \\ \hline\hline
	\begin{tabular}[c]{@{}c@{}}\textbf{Lower}\\ \textbf{Bound}\end{tabular} & Positive LP       &   indep. of $\epsilon$  &  $\Omega( \opt\cdot(h\Phi)^{-1}\ln(1/\epsilon))$ & 
	$\dist(\x{k},\Xstar)<\epsilon$          \\ \hline
	\end{tabular}
	\smallskip
	\caption{Convergence results for the discrete directed Physarum-inspired dynamics \eqref{eq:DdirD}.}
	\label{tableDiscrete}
\end{table}

We stated the bounds on $h$ in terms of the unknown quantities $\Phi$ and $\opt$. 
However, $\Phi/\opt \ge 1/C_3$ by Lemma~\ref{lem:fsDg} and hence replacing $\Phi/\opt$ 
by $1/C_3$ yields constructive bounds for $h$. Note that the upper bound on the step size 
does not depend on $\epsilon$ and that the bound on the number of iterations depends
\textit{logarithmically} on $1/\epsilon$ and \textit{quadratically} on $\opt/\Phi$.

What can be done if the initial point is not strongly dominating? 
For the transshipment problem it suffices to add an edge of high capacity 
and high cost from every source node to every sink node~\cite{Physarum-Complexity-Bounds,SV-Flow}. 
This will make the instance strongly dominating and will not affect the optimal solution. 
We generalize this observation to positive linear programs. We add an additional column 
equal to $b$ and give it sufficiently high capacity and cost. This guarantees that the 
resulting instance is strongly dominating and the optimal solution remains
unaffected. Moreover, our approach generalizes and improves upon~\cite[Theorem 1.2]{SV-Flow},
see Section~\ref{sec:Preconditioning}.

\paragraph*{Proof Techniques:}

The crux of the analysis in \cite{Ito-Convergence-Physarum,Physarum-Complexity-Bounds,SV-Flow} 
is to show that for large enough $k$, $\xk$ is close to a \emph{non-negative} flow
$\fk$ and then to argue that $\fk$ is close to an optimal flow $\fstr$. This line of 
arguments yields a convergence of $\xk$ to $\Xstar$ with a step size $h$ chosen independently of
$\epsilon$.

In Section~\ref{discretization}, we extend the preceding approach to positive linear programs, 
by generalizing the concept of non-negative cycle-free flows to 
non-negative \emph{feasible kernel-free} vectors (Subsection~\ref{subsection:CloseNonNengKernelFreeVector}).
Although, we use the same high level ideas as in~\cite{Physarum-Complexity-Bounds,SV-Flow}, 
we stress that our analysis generalizes all relevant lemmas 
in~\cite{Physarum-Complexity-Bounds,SV-Flow}
and it uses arguments from linear algebra and linear programming duality, 
instead of combinatorial arguments. 
Further, our core efficiency bounds (Subsection~\ref{subsec:subsecUsefulLemmas}) 
extend~\cite{SV-LP} and yield a~\emph{scale-invariant} determinant dependence 
of the step size and are applicable for any strongly dominating point (Subsection~\ref{subsec:SDCV}).

%% file: SimpleInstances.tex
\section{Convergence of the Physarum Dynamics: Simple Instances}\label{Simple Instances}\label{simple instances}

In this section, we prove Theorem~\ref{intro_thm} under the simplifying assumptions 
\eqref{positive c} to \eqref{x0 is dominating}, defined in page~\pageref{positive c}.

\subsection{Preliminaries}\label{preliminaries}

Note that we may assume that $A$ has full row-rank since any equation that is linearly dependent on other equations can be deleted without changing the feasible set. We continue to use $n$ and $m$ for the dimension of $A$. Thus, $A$ has rank $n$. We continue by fixing some terms and notation.
A \emph{basic feasible solution} of \eqref{ULP} is a pair of vectors $x$ and $f = (f_B, f_N)$, where $f_B = A_B^{-1} b$ and $A_B$ is a square $n \times n$ non-singular sub-matrix of $A$ and $f_N = 0$ is the vector indexed by the coordinates not in $B$, and $x=|f|$. Since $f$ uniquely determines $x$, we may drop the latter for the sake of brevity and call $f$ a basic feasible solution of \eqref{ULP}. A feasible solution $f$ is \emph{kernel-free} or \emph{non-circulatory}  if it is contained in the convex hull of the basic feasible solutions.\footnote{
	For the undirected shortest path problem, we drop the equation corresponding to the sink. Then $b$ becomes the negative indicator vector corresponding to the source node. Note that $n$ is one less than the number of nodes of the graph.
The basic feasible solutions are the simple undirected source-sink paths. A circulatory solution contains a cycle on which there is flow.}
We say that a vector $f'$ is \emph{sign-compatible} with a vector $f$ (of the same dimension) or $f$-sign-compatible if $f'_e \not= 0$ implies $f'_e f_e > 0$. In particular, $\supp(f') \subseteq \supp(f)$.
For a given capacity vector $x$ and a vector $f \in \R^m$ with $\supp(f) \subseteq \supp(x)$, we use $E(f) = \sum_e (c_e/x_e) f_e^2$ to denote the \emph{energy} of $f$. The energy of $f$ is infinite, if $\supp(f) \not\subseteq \supp(x)$. We use $\cost(f)= \sum_e c_e \abs{f_e} = c^{\rot} \abs{f}$ to denote the \emph{cost} of $f$.
Note that $E(x) = \sum_e (c_e/x_e) x_e^2 = \sum_e c_e x_e = \cost(x)$. 
We define the constants
$\cmax = \|c\|_\infty$ and $\cmin = \min_{e:c_e > 0} c_e$.

We use the following corollary of the finite basis theorem for polyhedra.

\begin{lemma}\label{sign-compatible representation}
	Let $f$ be a feasible solution of \eqref{ULP}. Then $f$ is the sum of a convex combination of at most $m$ basic feasible solutions plus a vector in the kernel of $A$. Moreover, all elements in this representation are sign-compatible with $f$.
\end{lemma}
\begin{proof} We may assume $f \ge 0$. Otherwise, we flip the sign of the appropriate columns of $A$. Thus, the system $Af =b,\ f \ge 0$ is feasible and $f$ is the sum of a convex combination of at most $m$ basic feasible solutions plus a vector in the kernel of $A$ by the finite basis theorem~\cite[Corollary 7.1b]{Schrijver99}. By definition, the elements in this representation are non-negative vectors and hence sign-compatible with $f$.
\end{proof}

\begin{lemma}[Gr\"onwall's Lemma]\label{Gronwalls Lemma} 
	Let $A, B, \alpha, \beta \in \R$, $\alpha \not= 0$, $\beta \not= 0$, and let $g$ be a continuous differentiable function on $[0,\infty)$. If $A + \alpha g(t) \le \dot{g}(t) \le B + \beta g(t)$ for all $t \ge 0$, then $- A/\alpha + (g(0) + A/\alpha) e^{\alpha t} \le g(t) \le  - B/\beta + (g(0) + B/\beta) e^{\beta t}$ for all $t \ge 0$. 
\end{lemma}
\begin{proof} 
We show the upper bound. Assume first that $B = 0$. Then 
\[  
	\frac{d}{dt} \frac{g}{e^{\beta t}} = 
	\frac{\dot{g} e^{\beta t} - \beta g e^{\beta t}  }{e^{2 \beta t}} 
	\le 0 \quad\text{implies}\quad  \frac{g(t)}{e^{\beta t}} 
	\le \frac{g(0)}{e^{\beta 0}} = g(0).
\]
If $B \not= 0$, define $h(t) = g(t) + B/\beta$. Then 
\[     
	\dot{h} = \dot{g} \le B + \beta g = B + \beta(h - B/\beta) 
	= \beta h
\]
and hence $h(t) \le h(0) e^{\beta t}$. 
Therefore $g(t) \le - B/\beta + (g(0) + B/\beta) e^{\beta t}$. 
\qedhere
\end{proof}

An immediate consequence of Gr\"onwall's Lemma is that the undirected Physarum dynamics~\eqref{PD}
initialized with any positive starting vector $x(0)$, generates a trajectory  
$\{x(t)\}_{t\geq0}$ such that each time state $x(t)$ is a positive vector.
Indeed, since $\dot{x}_e = | q_e | - x_e \ge - x_e$,
we have $x_e(t) \ge x_e(0)\cdot\exp\{-t\}$ for every index $e$ with $x_e(0)>0$ and every time $t$. 
Further, by \eqref{PD} and \eqref{Defq}, it holds for indices $e$ with $x_e(0) = 0$
that $x_e(t)=0$ for every time $t$.
Hence, the trajectory $\{x(t)\}_{t\geq0}$ has a time-invariant support.

\begin{lemma}[\cite{Johannson-Zou}]\label{Formula for q} 
	Let $R = \diag(c_e/x_e)$. Then $q = R^{-1} A^{\rot} p$, where $p = (A R^{-1} A^{\rot})^{-1} b$. 
\end{lemma}
\begin{proof} 
	$q$ minimizes $\sum_e r_e q_e^2$ subject to $A q = b$. 
	The Karush-Kuhn-Tucker (KKT) optimality conditions for 
	constrained optimization~\cite{Boyd04} imply the existence of 
	a vector $p$ such that $Rq = A^{\rot}p$. 
	Substituting into $A q = b$ yields $p = (A R^{-1} A^{\rot})^{-1} b$.
\end{proof}

\begin{lemma}\label{x stays dominating} 
	$\Xdom$ is an invariant set, i.e., if $x(0) \in \Xdom$ then $x(t) \in \Xdom$ for all $t$.
\end{lemma}
\begin{proof} Let $q(t)$ be the minimum energy feasible solution with respect to $R(t) = \diag(c_e/x_e(t))$, and let $f(t)$ be such that $f(0)$ is feasible, $\abs{f(0)} \le x(0)$, and $\dot{f}(t) = q(t) - f(t)$. Then
	$\frac{d}{dt} (Af - b) = A(q - f) = b - Af$ and hence $Af(t) - b = (Af(0) - b) e^{-t} = 0$. Thus $f(t)$ is feasible for all $t$. Moreover,
	\[ \frac{d}{dt} (f - x) = \dot{f} - \dot{x} = q - f - (\abs{q} - x) = q - \abs{q} - (f - x) \le - (f - x).\]
	Thus $f(t) - x(t) \le (f(0) - x(0)) e^{-t} \le 0$  by Gr\"onwall's Lemma applied with $g(t) = f(t) - x(t)$ and $\beta = -1$, and hence $f(t) \le x(t)$ for all $t$. Similarly,
	\[ \frac{d}{dt} (f + x) = \dot{f} + \dot{x} = q - f + (\abs{q} - x) = q + \abs{q} - (f + x) \ge - (f + x).\]
	Thus $f(t) + x(t) \ge (f(0) + x(0)) e^{-t} \ge 0$ by Gr\"onwall's Lemma applied with $g(t) = f(t) + x(t)$ and $\alpha = -1$ and $A = 0$, and hence $f(t) \ge - x(t)$ for all $t$.

	We conclude that $\abs{f(t)} \le x(t)$ for all $t$. Thus, $x(t) \in \Xdom$ for all $t$.
\end{proof}  

\subsection{The Convergence Proof}

We will first characterize the equilibrium points. They are precisely the points $\abs{f}$, where $f$ is a basic feasible solution; the proof uses assumption \eqref{distinctCost} in page~\pageref{positive c}. We then show that $E(x)$ is a Lyapunov function for~\eqref{PD}, in particular, $\dot{E} \le 0$ and $\dot{E} = 0$ if and only if $x$ is an equilibrium point. For this argument, we need that the energy of $q$ is at most the energy of $x$ with equality if and only if $x$ is an equilibrium point. This proof uses assumptions (\ref{positive c}) and (\ref{x0 is dominating}) in page~\pageref{positive c}.
It follows from the general theory of dynamical systems that $x(t)$ approaches an equilibrium point. Finally, we show that convergence to a non-optimal equilibrium is impossible.

\begin{lemma}[Generalization of Lemma 2.3 in~\cite{Bonifaci-Physarum}]\label{Lemma 2.3}
	Assume~\eqref{positive c}  to \eqref{x0 is dominating}.
	If $f$ is a basic feasible solution of~\eqref{ULP}, then $x = \abs{f}$ is an equilibrium point. Conversely, if $x$ is an equilibrium point, then $x = \abs{f}$ for some basic feasible solution $f$.
\end{lemma}
\begin{proof}
	Let $f$ be a basic feasible solution, let $x = \abs{f}$, and let $q$ be the minimum energy feasible solution with respect to the resistances $c_e/x_e$. We have $Aq = b$ and $\supp(q) \subseteq \supp(x)$ by definition of $q$. Since $f$ is a basic feasible solution there is a subset $B$ of size $n$ of the columns of $A$ such that $A_B$ is non-singular and $f = (A_B^{-1} b, 0)$. Since $\supp(q) \subseteq \supp(x) \subseteq B$, we have $q = (q_B,0)$ for some vector $q_B$. Thus, $b = Aq = A_B q_B$ and hence $q_B = f_B$. Therefore $\dot{x} = \abs{q} - x = 0$ and $x$ is an equilibrium point.

	Conversely, if $x$ is an equilibrium point, $\abs{q_e} = x_e$ for every $e$. By changing the signs of some columns of $A$, we may assume $q \ge 0$. Then $q = x$. Since $q_e = x_e/c_e A_e^{\rot} p$ where $A_e$ is the $e$-th column of $A$ by Lemma~\ref{Formula for q}, we have $c_e = A_e^{\rot} p$, whenever $x_e > 0$. By Lemma~\ref{sign-compatible representation}, $q$ is a convex combination of basic feasible solutions and a vector in the kernel of $A$ that are sign-compatible with $q$. The vector in the kernel must be zero as $q$ is a minimum energy feasible solution. For any basic feasible solution $z$ contributing to $q$, we have $\supp(z) \subseteq \supp(x)$.  Summing over the $e \in \supp(z)$, we obtain $\cost(z) = \sum_{e \in \supp(z)} c_e z_e = \sum_{e \in \supp(z)} z_e A_e^{\rot} p = b^{\rot} p$. Thus, the convex combination involves only a single basic feasible solution by assumption~\eqref{distinctCost} and hence $x$ is a basic feasible solution.
\end{proof}

The vector $x(t)$ dominates a feasible solution at all times. Since $q(t)$ is the minimum energy feasible solution at time $t$, this implies $E(q(t)) \le E(x(t))$ at all times. A further argument shows that we have equality if and only if $x = \abs{q}$.

\begin{lemma}[Generalization of Lemma 3.1 in ~\cite{Bonifaci-Physarum}]\label{energy less than cost}
	Assume~\eqref{positive c}  to \eqref{x0 is dominating}. At all times, $E(q) \le E(x)$. If $E(q) = E(x)$, then $x = \abs{q}$.
\end{lemma}
\begin{proof} Recall that $x(t) \in \Xdom$ for all $t$. Thus, at all times, there is a feasible $f$ such that $\abs{f} \le x$. Since $q$ is a minimum energy feasible solution, we have
\[  
E(q) \le E(f) \le E(x).
\]
If $E(q) = E(x)$ then $E(q) = E(f)$ and hence $q = f$ since the minimum energy feasible solution is unique. Also, $\abs{f} = x$ since $\abs{f} \le x$ and $\abs{f_e} < x_e$ for some $e$ implies $E(f) < E(x)$. The last conclusion uses $c > 0$.
\end{proof}

Lyapunov functions are the main tool for proving convergence of dynamical systems. We show that $E(x)$ is a Lyapunov function for~\eqref{PD}.

\begin{lemma}[Generalization of Lemma 3.2 in~\cite{Bonifaci-Physarum}]\label{Lemma 3.2}
	Assume~\eqref{positive c}  to \eqref{x0 is dominating}. $E(x)$ is a Lyapunov function for \eqref{PD}, i.e., it is continuous as a function of $x$, $E(x) \ge 0$, $\dot{E}(x) \le 0$ and $\dot{E}(x) = 0$ if and only if $\dot{x} = 0$.
\end{lemma}
\begin{proof}
	$E$ is clearly continuous and non-negative. Recall that $E(x) = \cost(x)$. Let $R$ be the diagonal matrix with entries $c_e/x_e$. Then
	\begin{align*}
		\frac{d}{dt}{\cost(x)} & = c^{\rot} (\abs{q} - x)                           & \text{by~\eqref{PD}}                        \\
		                       & = x^{\rot} R \abs{q} - x^{\rot} R x                     & \text{since $c = Rx$}                       \\
		& = x^{\rot} R^{1/2} R^{1/2} \abs {q} - x^{\rot} R x \\
		                       & \le (q^{\rot} R q)^{1/2} (x^{\rot} R x)^{1/2} - x^{\rot} R x & \text{by Cauchy-Schwarz}                    \\
		                       & \le (x^{\rot} R x)^{1/2} (x^{\rot} R x)^{1/2} - x^{\rot} R x & \text{by Lemma~\ref{energy less than cost}} \\
		&= 0.
	\end{align*}
	Observe that $\frac{d}{dt}{\cost(x)} = 0$ implies that both inequalities above are equalities. This is only possible if the vectors $\abs{q}$ and $x$ are parallel and $E(q) = E(x)$. Thus, $x = \abs{q}$ by Lemma~\ref{energy less than cost}.
\end{proof}

It follows now from the general theory of dynamical systems that $x(t)$ converges to an equilibrium point. 

\begin{corollary}[Generalization of Corollary 3.3. in~\cite{Bonifaci-Physarum}.]
	Assume~\eqref{positive c}  to \eqref{x0 is dominating}. As $t \rightarrow \infty$, $x(t)$ approaches an equilibrium point and $c^{\rot} x(t)$ approaches the cost of the corresponding basic feasible solution.
\end{corollary}
\begin{proof}
	The proof in~\cite{Bonifaci-Physarum} carries over. We include it for completeness. The existence of a Lyapunov function $E$ implies by \cite[Corollary 2.6.5]{LaSalle} that $x(t)$  approaches the set $\set{x \in \R_{\ge 0}^m}{\dot{E} = 0}$, which by Lemma~\ref{Lemma 3.2} is the same as the set $\set{x \in \R_{\ge 0}^m}{\dot{x} = 0}$. Since this set consists of isolated points (Lemma~\ref{Lemma 2.3}), $x(t)$ must approach one of those points, say the point $x_0$. When $x = x_0$ , one has $E(q) = E(x) = \cost(x) = c^{\rot} x$.
\end{proof}

It remains to exclude that $x(t)$ converges to a nonoptimal equilibrium point.

\begin{theorem}[Generalization of Theorem 3.4 in~\cite{Bonifaci-Physarum}]
	Assume~\eqref{positive c}  to \eqref{x0 is dominating}. 
	As $t \rightarrow \infty$, $c^{\rot} x(t)$ converges to the cost of the optimal solution and $x(t)$ converges to the optimal solution.
\end{theorem}
\begin{proof}
	By the corollary, it suffices to prove the second part of the claim.
	For the second part, assume that $x(t)$ converges to a non-optimal solution $z$. Let $x^*$ be the optimal solution and
	let $W = \sum_e x^*_e c_e \ln x_e$. Let $\delta = (\cost(z) - \cost(x^*))/2$. 
	Note that for all sufficiently large $t$, we have 
	$E(q(t)) \ge \cost(z) - \delta \ge \cost(x^*) + \delta$. 
	Further, by definition $q_e = (x_e/c_e) A_e^{\rot} p$ and thus
	\[
		\dot{W} = \sum_e x^*_e c_e \frac{\abs{q_e} - x_e}{x_e} 
		= \sum_e x^*_e \abs{A_e^{\rot} p} - \cost(x^*)
		\geq \sum_e x^*_e A_e^{\rot} p - \cost(x^*) \geq \delta,
	\]
	where the last inequality follows by 
	$\sum_e x^*_e A_e^{\rot} p = b^{\rot} p = E(q)\geq \cost(x^*) + \delta$.
	Hence $W \rightarrow \infty$, a contradiction to the fact that $x$ is bounded.
\end{proof}

%% file: undirected.tex
\section{Convergence of the Physarum Dynamics: General Instances}\label{convergence undirected}\label{General Case}\label{general instances}

In this section, we prove Theorem~\ref{intro_thm} under the more general assumptions
	\eqref{c is nonnegative} to \eqref{positive start vector}, 
	defined in page~\pageref{c is nonnegative}.

\subsection{Existence of a Solution with Domain $[0,\infty)$}

In this subsection we show that a solution $x(t)$ to~\eqref{PD} has domain $[0,\infty)$. We first derive an explicit formula for the minimum energy feasible solution $q$ and then show that the mapping 
$x \mapsto q$ is Lipschitz continuous; this implies existence of a solution with domain $[0,\infty)$ by standard arguments.

\subsubsection{The Minimum Energy Solution}\label{minimum energy solution}

Recall that for $\gamma_{A}=\gcd(\{A_{ij}:A_{ij}\neq0\})\in\Z_{>0}$, we defined by
\[
\D = \max\left\{ \left|\det\left(M\right)\right|\, :\, 
M\text{ is a square submatrix of }A/\gamma_{A}
\text{ with dimension }n-1\text{ or }n\right\}.
\]

We derive now properties of the minimum energy solution. In particular, if every non-zero vector in the kernel of $A$ has positive cost, 
\begin{compactenum}[\mbox{}\hspace{\parindent}(i)]
\item the minimum energy feasible solution is kernel-free and unique (Lemma~\ref{q is kernel-free}),
\item $\abs{q_e} \le \D\onenorm{b/\gamma_{A}}$ for every $e \in [m]$ (Lemma~\ref{upper bound on q}), 
\item $q$ is defined by~\lref{explicit pq} (Lemma~\ref{formula for $q$, general case}), and
\item $E(q) = b^{\rot} p$, where $p$ is defined by~\lref{explicit pq} (Lemma~\ref{E(q)=bTp general case}).
\end{compactenum}
We note that for positive cost vector $c > 0$, these results are known.

We proceed by establishing some useful properties on basic feasible solutions.

\begin{lemma}\label{lem:fsDg}
	Suppose $A\in\mathbb{Z}^{n\times m}$ is an integral matrix, and $b \in \Z^n$ is an integral vector.
	Then, for any basic feasible solutions $f$ with $Af=b$ and $f\ge 0$, it holds that 
	$\infnorm{f}\leq \D \onenorm{b/\gamma_{A}}$ and 
	$f_{j}\neq0$ implies $|f_{j}|\geq1/(\D \gamma_{A})$.
\end{lemma}
\begin{proof}
	Since $f$ is a basic feasible solution, it has the form $f=\left(f_{B},0\right)$
	such that $A_{B}\cdot f_{B}=b$ where $A_{B}\in\mathbb{Z}^{n\times n}$ is an 
	invertible submatrix of $A$.
	We write $M_{-i,-j}$ to denote the matrix $M$ with deleted $i$-th row 
	and $j$-th column. Let $Q_j$ be the matrix formed by replacing the $j$-th column 
	of $A_{B}$ by the column vector $b$.
	Then, using the fact that $\det(tA)=t^{n}\det(A)$ for every $A\in\mathbb{R}^{n\times n}$ 
	and $t\in\mathbb{Z}$, Cramer's rule yields
	\[
	\left|f_{B}(j)\right|
	=
	\left|\frac{\det\left(Q_{j}\right)}{\det\left(A_{B}\right)}\right|
	=
	\frac{1}{\gamma_{A}}\left|\sum_{k=1}^{n}\frac{\left(-1\right)^{j+k}\cdot b_{k}\cdot\det\left(\gamma_{A}^{-1}[A_{B}]_{-k,-j}\right)}
	{\det\left(\gamma_{A}^{-1}A_{B}\right)}\right|
	\]
	By the choice of $\gamma_{A}$, the values $\det(\gamma_{A}^{-1}A_{B})$ and 
	$\det(\gamma_{A}^{-1}[A_{B}]_{-k,-j})$ are integral for all $k$, it follows that
	\[
	\left|f_{B}(j)\right|\leq \D\onenorm{b/\gamma_{A}} \qquad \text{and} \qquad f_B(j) \ne 0
	\quad\Longrightarrow\quad \frac{1}{\D \gamma_{A}}\leq\left|f_{B}(j)\right|.\hfill\qedhere
	\]
\end{proof}

\begin{lemma}\label{q is kernel-free}
	If every non-zero vector in the kernel of $A$ has positive cost, the minimum energy feasible solution is kernel-free and unique.
\end{lemma}
\begin{proof} Let $q$ be a minimum energy feasible solution. Since $q$ is feasible, it can be written as $q_n + q_r$, where $q_n$ is a convex combination of basic feasible solutions and $q_r$ lies in the kernel of $A$. Moreover, all elements in this representation are sign-compatible with $q$ by Lemma~\ref{sign-compatible representation}. If $q_r \not= 0$, the vector
	$q - q_r$ is feasible and has smaller energy, a contradiction. Thus $q_r = 0$.

	We next prove uniqueness. Assume for the sake of a contradiction that there are two distinct minimum energy feasible solutions $\qone$ and $\qtwo$. We show that the solution $(\qone + \qtwo)/2$ uses less energy than $\qone$ and $\qtwo$. Since $h \mapsto h^2$ is a strictly convex function from $\R$ to $\R$, the average of the two solutions will be better than either solution if there is 
	an index $e$ with $r_e > 0$ and $\qone_e \not= \qtwo_e$. The difference
	$z = \qone - \qtwo$ lies in the kernel of $A$ and hence $\cost(z) = \sum_e c_e \abs{z_e} > 0$. Thus there is an $e$ with $c_e > 0$ and $z_e \not= 0$. We have now shown uniqueness.
\end{proof}

\begin{lemma}\label{upper bound on q} Assume that every non-zero vector in the kernel of $A$ has positive cost. Let $q$ be the minimum energy feasible solution. Then $\abs{q_e}  \le \D\onenorm{b/\gamma_{A}}$ for every $e$.
\end{lemma}
\begin{proof} 
	Since $q$ is a convex combination of basic feasible solutions, 
	$\abs{q_e} \le \max_z \abs{z_e}$ where $z$ ranges over basic feasible solutions
	of the form $(z_B,0)$, where $z_B = A_B^{-1} b$ and 
	$A_B\in\mathbb{R}^{n \times n}$ is a non-singular submatrix of $A$.
	Thus, by Lemma~\ref{lem:fsDg} every component of $z$ is bounded by $\D\onenorm{b/\gamma_{A}}$.
\end{proof}

In~\cite{SV-LP}, the bound $|q_e|\le \D^2 m \onenorm{b}$ was shown.
We will now derive explicit formulae for the minimum energy solution $q$. We will express $q$ in terms of a vector $p \in \R^n$, which we refer to as the \emph{potential}, by analogy with the network setting, in which $p$  can be interpreted as the electric potential of the nodes. 
The energy of the minimum energy solution is equal to $b^{\rot} p$.
We show that the mapping $x\mapsto q$ is locally Lipschitz.
Note that for $c> 0$ these facts are well-known.
Let us split the column indices $[m]$ of $A$ into
\begin{equation}\label{eq:idxSetPZ}
P := \set{e\in[m]}{c_e > 0}\quad\text{ and }\quad Z := \set{e\in[m]}{c_e = 0}.
\end{equation}

\begin{lemma}\label{formula for $q$, general case}
	Assume that every non-zero vector in the kernel of $A$ has positive cost. Let $r_e = c_e/x_e$ and let $R$ denote the corresponding diagonal matrix. Let us split $A$ into $A_P$ and $A_Z$, and $q$ into $q_P$ and $q_Z$.  Since $A_Z$ has linearly independent columns, we may assume that the first $\abs{Z}$ rows of $A_Z$ form a square non-singular matrix. We can thus write $A=\big[\begin{smallmatrix}A'_P & A'_Z \\ A''_P & A''_Z \end{smallmatrix}\big]$ with invertible $A'_Z$. Then the minimum energy solution satisfies 
	\begin{equation}\label{both blocks}
		\begin{bmatrix} A'_P & A'_Z  \\ A''_P & A''_Z \end{bmatrix}
		\begin{bmatrix} q_P \\ q_{Z} \end{bmatrix}
		= \begin{bmatrix} b' \\ b'' \end{bmatrix} 
		\quad\text{ and }\quad
		\begin{bmatrix} R_P &  0 \\ 0 & 0 \end{bmatrix}
		\begin{bmatrix} q_{P} \\ q_{Z} \end{bmatrix}  =
		\begin{bmatrix} {A'_P}^{\rot} & {A''_P}^{\rot}\\ {A'_Z}^{\rot} & {A''_Z}^{\rot} \end{bmatrix}  \begin{bmatrix} p' \\ p'' \end{bmatrix}
\end{equation}
for some vector $p = \big[ \begin{smallmatrix} p' \\ p'' \end{smallmatrix} \big]$; here $p'$ has dimension $\abs{Z}$. The equation system~\eqref{both blocks} has a unique solution given by 
	\begin{align}\label{explicit pq}
		\begin{bmatrix} q_Z                         \\q_P\end{bmatrix}
		=\begin{bmatrix}[A'_Z]^{-1} (b' - A'_P q_P) \\ R_P^{-1} A^{\rot}_P p\end{bmatrix}
		\quad \text{ and }\quad
		\begin{bmatrix} p'                          \\p''\end{bmatrix}
		=\begin{bmatrix}- [[A'_Z]^{\rot}]^{-1} [A''_Z]^{\rot} p'' \\MR^{-1}M^{\rot} (b'' - A''_Z [A'_Z]^{-1} b')\end{bmatrix},
	\end{align}
	where $M= A''_P  - A''_Z [A'_Z]^{-1} A'_P$ is the Schur complement of the block $A'_Z$ of the matrix $A$.
\end{lemma}

\begin{proof}  
	$q$ minimizes $E(f) = f^{\rot} R f$ among all solutions of $Af = b$. 
	The KKT	conditions state that $q$ must satisfy $R q = A^{\rot} p$ for some $p$. Note that $2Rf$ is the gradient of the energy $E(f)$ with respect to $f$ and that the $-A^{\rot}p$ is the gradient of $p^{\rot}(b - A f)$ with respect to $f$. We may absorb the factor $-2$ in $p$. Thus $q$ satisfies \eqref{both blocks}. 

	We show next that the linear system \eqref{both blocks} has a unique solution. The top $\abs{Z}$ rows of the left system in~\eqref{both blocks} give
	\begin{equation}\label{expression for qZ}
		q_Z = [A'_Z]^{-1} (b' - A'_P q_P).
	\end{equation}
	Substituting this expression for $q_Z$ into the bottom $n - \abs{Z}$ rows of the left system in~\eqref{both blocks} yields
	\[ M  q_P = b'' - A''_Z [A'_Z]^{-1} b'. \]
	From the top $\abs{P}$ rows of the right system in~\eqref{both blocks} we infer
	$q_P = R_P^{-1} A^{\rot}_P \cdot p$.
	Thus
	\begin{equation}\label{equation for p}
		M R_P^{-1} A^{\rot}_P \cdot p = b'' - A''_Z [A'_Z]^{-1} \cdot b'.
	\end{equation}
	The bottom $n - \abs{Z}$ rows of the right system in~\eqref{both blocks} yield 
	$0 = A_Z^{\rot} p = [A'_Z]^{\rot} p' + [A''_Z]^{\rot} p''$ and hence
	\begin{equation}\label{expression for pprime}
		p' = - [[A'_Z]^{\rot}]^{-1} [A''_Z]^{\rot} p''.
	\end{equation}
	Substituting~(\ref{expression for pprime}) into~\eqref{equation for p} yields
	\begin{align}
		b'' - A''_Z [A'_Z]^{-1} b' & =  MR_{P}^{-1}\left([A'_{P}]^{\rot}p'+[A''_{P}]^{\rot}p''\right) \notag
		\\ \nonumber
		                           & = MR_{P}^{-1}\left([A''_{P}]^{\rot}-
		                           [A'_{P}]^{\rot}[[A'_{Z}]^{\rot}]^{-1}[A''_{Z}]^{\rot}\right) p'' \\
		                           & = MR_P^{-1}M^{\rot} p''. \label{expression for pprimeprime}
	\end{align}
	It remains to show that the matrix $MR_P^{-1}M^{\rot}$ is non-singular. We first observe that the rows of $M$ are linearly independent. Consider the left system in \eqref{both blocks}. Multiplying the first $\abs{Z}$ rows by $[A'_Z]^{-1}$ and then subtracting $A''_Z$ times the resulting rows from the last $n - \abs{Z}$ rows turns $A$ into the matrix
	$Q=\big[\begin{smallmatrix}
		[A'_Z]^{-1} A'_P & I  \\ M & 0
	\end{smallmatrix}\big].$
	By assumption, $A$ has independent rows. Moreover, the preceding operations guarantee that $\rank(A)=\rank(Q)$. Therefore, $M$ has independent rows. Since $R_P^{-1}$ is a positive diagonal matrix, $R_P^{-1/2}$ exists and is a positive diagonal matrix. Let $z$ be an arbitrary non-zero vector of dimension $\abs{P}$. Then $z^{\rot} MR_P^{-1}M^{\rot}z = (R_P^{-1/2} M^{\rot} z)^{\rot} (R_P^{-1/2} M^{\rot} z) > 0$ and hence  $MR_P^{-1}M^{\rot}$ is non-singular. It is even positive semi-definite.

	There is a shorter proof that the system~\eqref{both blocks} has a unique solution. However, the argument does not give an explicit expression for the solution. In  the case of a convex objective function and affine constraints, the KKT conditions are sufficient for being a global minimum. Thus any solution to~\eqref{both blocks} is a global optimum. We have already shown in Lemma~\ref{q is kernel-free} that the global minimum is unique.
\end{proof}

We next observe that the energy of $q$ can be expressed in terms of the potential. 
\begin{lemma}\label{E(q)=bTp general case} 
	Let $q$ be the minimum energy feasible solution and let $f$ be any feasible solution. 
	Then $E(q) = b^{\rot} p = f^{\rot} A^{\rot} p$. 
\end{lemma}
\begin{proof}
	As in the proof of Lemma~\ref{formula for $q$, general case}, we split $q$ into $q_P$ and $q_Z$, $R$ into $R_P$ and $R_Z$, and $A$ into $A_P$ and $A_Z$. Then
	\begin{align*}
		E(q) & = q_P^{\rot} R_P q_P     & \text{by the definition of $E(q)$ and since $R_Z = 0$}                \\
		     & = p^{\rot} A_P q_p       & \text{by  the right system in~\eqref{both blocks}}  \\
		     & = p^{\rot} (b - A_Z q_Z) & \text{by the left system in~\eqref{both blocks}}                                  \\
		     & = b^{\rot} p             & \text{by the right system in~\eqref{both blocks}}.
	\end{align*}
	For any feasible solution $f$, we have $f^{\rot} A^{\rot} p = b^{\rot} p$.
\end{proof}

\subsubsection{The Mapping $x \mapsto q$ is Locally Lipschitz} 

We show that the mapping $x \mapsto q$ is locally Lipschitz continuous; this implies existence of a solution $x(t)$ with domain $[0,\infty)$ by standard arguments. 
Our analysis builds upon Cramer's rule and the Cauchy-Binet formula. The Cauchy-Binet formula extends Kirchhoff's spanning tree theorem which was used in ~\cite{Physarum} for the analysis of the undirected shortest path problem.

\begin{lemma}[Local Lipschitz Condition]\label{Locally Lipschitz, General Case}
	Assume $c \ge 0$, no non-zero vector in the kernel of $A$ has cost zero, and that $A$, $b$, and $c$ are integral. Let $\alpha,\beta > 0$. For any two vectors $x$ and $\tilde{x}$ in $\R^m$ with 
	$\alpha \le x_e,\tilde{x}_e \le \beta$ for all $e$, 
	define $\gamma := 2 m^n (\beta/\alpha)^n \cmax^n \D^2 \onenorm{b/\gamma_{A}}$. 
	Then $\big|\abs{q_e(x)} - \abs{q_e(\tilde{x})}\big| \le \gamma \infnorm{x - \tilde{x}}$ 
	for every $e \in [m]$. 
\end{lemma}
\begin{proof}
	First assume that $c>0$. By Cramer's rule
	\[
		(A R^{-1} A^{\rot})^{-1} = \frac{1}{\det (A R^{-1} A^{\rot})} ( (-1)^{i+j} \det (M_{-j,-i}))_{ij},
	\]
	where $M_{-i,-j}$ is obtained from $A R^{-1} A^{\rot}$ by deleting the $i$-th row and the $j$-th column. For a subset $S$ of $[m]$ and an index $i \in [n]$, let $A_S$ be the $n \times \abs{S}$ matrix consisting of the columns selected by $S$ and let $A_{-i,S}$ be the matrix obtained from $A_S$ by deleting row $i$. If $D$ is a diagonal matrix of size $m$, then $(A D)_S = A_S D_S$. The Cauchy-Binet theorem expresses the determinant of a product of two matrices (not necessarily square) as a sum of determinants of square matrices. It yields
	\begin{align*}
		\det (A R^{-1} A^{\rot})
		= \sum_{S \subseteq [m];\ \abs{S} = n} (\det( (AR^{-1/2})_S))^2\\
		= \sum_{S \subseteq [m];\ \abs{S} = n} (\prod_{e \in S} x_e/c_e) \cdot  (\det A_S)^2.
	\end{align*}
	Similarly,
	\begin{align*}
		\det (A R^{-1} A^{\rot})_{-i,-j}
		  & = \sum_{S \subseteq [m];\ \abs{S} = n-1} (\prod_{e \in S} x_e/c_e) \cdot  (\det A_{-i,S} \cdot \det A_{-j,S}).
	\end{align*}
Using $p = (A R^{-1} A^{\rot})^{-1} b$, we obtain
	\begin{equation}\label{explicit p}
		p_i = \frac{\sum_{j \in [n]} (-1)^{i + j}\sum_{S \subseteq [m];\ \abs{S} = n-1} (\prod_{e \in S} x_e/c_e) \cdot  (\det A_{-i,S} \cdot \det A_{-j,S}) b_j }{\sum_{S \subseteq [m];\ \abs{S} = n} (\prod_{e \in S} x_e/c_e) \cdot  (\det A_S)^2}.
	\end{equation}
Substituting into $q = R^{-1} A^{\rot} p$ yields
	\begin{align}\label{explicit q}
		q_e
		  & = \frac{x_e}{c_e} A_e^{\rot} p                                                                                                                                                                                                                                                                                               \notag \\
		  & = \frac{x_e}{c_e} \sum_i A_{i,e} \cdot \frac{\sum_{j \in [n]} (-1)^{i + j +2n}\sum_{S \subseteq [m];\ \abs{S} = n-1} (\prod_{e' \in S} x_{e'}/c_{e'}) \cdot  (\det A_{-i,S} \cdot \det A_{-j,S}) b_j }{\sum_{S \subseteq [m];\ \abs{S} = n} (\prod_{e' \in S} x_{e'}/c_{e'}) \cdot  (\det A_S)^2} \notag                \\
		  & = \frac{\sum_{S\subseteq[m];\ \abs{S}=n-1}(\prod_{e'\in S\cup e}x_{e'}/c_{e'})\cdot\sum_{i\in[n]}(-1)^{i+n}A_{i,e}\det A_{-i,S}\cdot\sum_{j\in[n]}(-1)^{j+n}b_{j}\det A_{-j,S}}{\sum_{S\subseteq[m];\ \abs{S}=n}(\prod_{e'\in S}x_{e'}/c_{e'})\cdot(\det A_{S})^{2}} \notag \\
		  & = \frac{\sum_{S \subseteq [m];\ \abs{S} = n-1} (\prod_{e' \in S\cup e} x_{e'}/c_{e'}) \cdot  \det (A_S | A_e) \cdot \det (A_S | b) }{\sum_{S \subseteq [m];\ \abs{S} = n} (\prod_{e' \in S} x_{e'}/c_{e'}) \cdot  (\det A_S)^2},
	\end{align}
	where $(A_S| A_e)$, respectively $(A_S | b)$, denotes the $n \times n$ matrix whose columns are selected from $A$ by $S$ and whose last column is equal to $A_e$, respectively $b$. 

	We are now ready to estimate the derivative $\partial q_e/\partial x_i$. Assume first that $e \not= i$. By the above,
	$q_e = \frac{x_e}{c_e} \frac{F + G x_i/c_i}{H + I x_i/c_i}$,
	where $F$, $G$, $H$ and $I$ are given implicitly by \eqref{explicit q}. Then
	\[    
		\left|\frac{\partial q_{e}}{\partial x_{i}}\right|=
		\left|\frac{x_{e}}{c_{e}}\cdot\frac{FI/c_{i}-GH/c_{i}}{(H+Ix_{i}/c_{i})^{2}}\right|
		\le\frac{2\cdot{\binom{m}{n-1}}\beta^{n}\D^{2}\onenorm{b/\gamma_{A}}}{(\alpha/\cmax)^{n}}
		\le\gamma.
	\]
	For $e = i$, we have
	$q_e = \frac{G x_e/c_e}{H + I x_e/c_e}$,
	where $G$, $H$, and $I$ are given implicitly by \eqref{explicit q}. Then
	\[    
		\left|\frac{\partial q_{e}}{\partial x_{e}}\right|=
		\left|\frac{GH/c_{e}}{(H+Ix_{e}/c_{e})^{2}}\right|
		\le\frac{{\binom{m}{n-1}}\beta^{n}\D^{2}\onenorm{b/\gamma_{A}}}{(\alpha/\cmax)^{n}}
		\le\gamma.
	\]
	Finally, consider $x$ and $\tilde{x}$ with $\alpha \le x_e, \tilde{x}_e \le \beta$ for all $e$. 
	Let $\bar{x}_\ell = (\tilde{x}_1, \ldots, \tilde{x}_\ell, x_{\ell + 1}, \ldots, x_m)$. Then
	\[ \abs{\abs{q_e(x)} - \abs{q_e(\tilde{x})}} \le \abs{q_e(x) - q_e(\tilde{x})} \le \sum_{0 \le \ell < m} \abs{q_e(\bar{x}_\ell) - q_e(\bar{x}_{\ell + 1})} \le \gamma \onenorm{x - \tilde{x}}.
	\]
	In the general case where $c\ge 0$, we first derive an expression for $p''$ similar to~\eqref{explicit p}.
	Then the equations for $p'$ in~\eqref{explicit pq} yield $p'$, the equations for $q_P$ in~\eqref{explicit pq} yield $q_P$, and finally the equations for $q_Z$ in~\eqref{explicit pq} yield $q_Z$.
\end{proof}

We are now ready to establish the existence of a solution with domain $[0,\infty)$. 

\begin{lemma}\label{x is bounded}\label{existence}
	The solution to the undirected dynamics in~\eqref{PD} has domain $[0,\infty)$. 
	Moreover, for every $t\geq0$ and $e\in[m]$, we have
	\[
		x_{e}(0)\cdot\exp\{-t\}\le x_{e}(t)\le
		\D\onenorm{b/\gamma_{A}}+\max(0,x_{e}(0)-\D\onenorm{b/\gamma_{A}})\cdot\exp\{-t\}.
	\]
\end{lemma}
\begin{proof}
	Consider any $x_0 > 0$ and any $t_0 \ge 0$. We first show that there is a positive $\delta'$ (depending on $x_0$) such that a unique solution $x(t)$ with $x(t_0) = x_0$ exists for $t \in (t_0 - \delta',t_0 + \delta')$. By the Picard-Lindel\"{o}f Theorem~\cite[Theorem 2.2]{Teschl12},
	this holds true if the mapping $x \mapsto \abs{q} - x$ is continuous and satisfies a Lipschitz condition in a neighborhood of $x_0$. Continuity clearly holds. Let $\epsilon = \min_i (x_0)_i/2$ and let $U = \set{x}{\infnorm{x - x_0} < \epsilon}$. Then for every $x, \tilde{x} \in U$ and every $e$
	\[
		\big|\abs{q_e(x)} - \abs{q_e(\tilde{x})}\big| \le \gamma \onenorm{x - \tilde{x}},
	\]
	where $\gamma$ is as in Lemma~\ref{Locally Lipschitz, General Case}.
	Local existence implies the existence of a solution which cannot be extended. Since $q$ is bounded (Lemma~\ref{upper bound on q}), $x$ is bounded at all finite times, and hence the solution exists for all $t$.
	The lower bound $x_e(t)  \ge x_e(0)\cdot \exp\{-t\} > 0$ for all $e$, 
	holds by Lemma~\ref{Gronwalls Lemma} 
	with $A = 0$ and $\alpha = -1$.
	Since $\abs{q_e} \le \D\onenorm{b/\gamma_{A}}$, 
	$\dot{x}_e = \abs{q_e} - x_e \le \D\onenorm{b/\gamma_{A}} - x_e$, 
	we have $x_e(t) \le \D\onenorm{b/\gamma_{A}} + \max(0, x_e(0) - 
	\D\onenorm{b/\gamma_{A}})\cdot \exp\{-t\}$ 
	by Lemma~\ref{Gronwalls Lemma} with $B =  \D\onenorm{b/\gamma_{A}}$ and $\beta = -1$.
\end{proof}

\subsection{LP Duality}\label{subsec:maxflowmincut}

The energy $E(x)$ is no longer a Lyapunov function, e.g., if $x(0) \approx \bfzero$, $x(t)$ and hence $E(x(t))$ will grow initially. We will show that energy suitably scaled is a Lyapunov function. What is the appropriate scaling factor? In the case of the undirected shortest path problem, \cite{Physarum} used the minimum capacity of any source-sink cut as a scaling factor. The proper generalization to our situation is to consider the linear program $\max\{\alpha : Af = \alpha b, |f| \le x\}$, where $x$ is a fixed positive vector. Linear programming duality yields the corresponding minimization problem which generalizes the minimum cut problem to our situation. 

\begin{lemma}\label{dual}
    Let $x \in \R_{>0}^{m}$ and $b \not= 0$. The linear programs 
	\begin{align}\label{LP-Duality}
		\max\{\alpha : Af = \alpha b, |f| \le x\}
		\qquad\text{and}\qquad \min \{|y^{\rot} A| x: b^{\rot}y = -1\}
	\end{align}
are feasible and have the same objective value. 
	Moreover, there is a finite set $\Y_A = \sset{d^1,\ldots,d^K}$ of vectors $d^i \in \R_{\ge 0}^m$ that are independent of $x$ such that the minimum above is equal to $\Copt = \min_{d \in \Y_{A}} d^{\rot}x$. There is a feasible $f$ with $\abs{f} \le x/\Copt$.
	\footnote{
		In the undirected shortest path problem, the $d$'s are the incidence vectors of the undirected source-sink cuts. Let $S$ be any set of vertices containing $s_0$ but not $s_1$, and let $\mathbf{1}^S$ be its associated indicator vector. The cut corresponding to $S$ contains the edges having exactly one endpoint in $S$. Its  indicator vector is  $d^S=\abs{A^{\rot} \mathbf{1}^S}$. Then $d^S_e = 1$ iff $\abs{S \cap \sset{u,v}} = 1$, where $e = (u,v)$ or $e = (v,u)$, and $d^S_e=0$ otherwise. For a vector $x \ge 0$, $(d^S)^{\rot} x$ is the capacity of the source-sink cut $(S,V\backslash S)$. In this setting, $\Copt$ is the value of a minimum cut.}
\end{lemma}
\begin{proof}
	\providecommand{\zn}{z^{-}}
	\providecommand{\zp}{z^{+}}
	The pair $(\alpha, f) = (0,0)$ is a feasible solution for the maximization problem. Since $b\neq 0$, there exists $y$ with $ b^{\rot}y = -1$ and thus both problems are feasible. The dual of $\max\{\alpha:Af - \alpha b = 0,f \le x, -f \le x\}$
	has unconstrained variables $y \in \R^n$ and non-negative variables $\zp,\zn\in\R^m$ and reads
	\begin{equation} \label{dual LP}
		\min\{x^{\rot} (\zp + \zn) : - b^{\rot} y = 1, A^{\rot}y + \zp - \zn = 0,\  \zp,\zn\ge 0\}.
	\end{equation}
	From $\zn = A^{\rot} y + \zp$, $\zp \ge 0$, $\zn \ge 0$ and $x > 0$, we conclude $\min(\zp,\zn) = 0$ in an optimal solution. Thus $\zn= \max(0, A^{\rot}y)$  and $\zp= \max(0, -A^{\rot} y)$ and hence  $\zp + \zn= \abs{A^{\rot} y}$ in an optimal dual solution. Therefore, \eqref{dual LP} and the right LP in \eqref{LP-Duality} have the same objective value.

	We next show that the dual attains its minimum at a vertex of the feasible set. For this it suffices to show that its feasible set contains no line. Assume it does. Then there are vectors  $d = (y_1,\zp_1,\zn_1)$, $d$ non-zero, and $p = (y_0,\zp_0,\zn_0)$ such that $(y, \zp,\zn) = p + \lambda d = (y_0 + \lambda y_1, \zp_0 + \lambda \zp_1, \zn_0 + \lambda \zn_1)$ is feasible for all $\lambda \in \R$. Thus $\zp_1 = \zn_1 = 0$. Note that if either $\zp_1$ or $\zn_1$ were non-zero then either  $\zp_0 + \lambda \zp_1$ or $\zn_0 + \lambda \zn_1$ would have a negative component for some $\lambda$. Then $A^{\rot}y + \zp + \zn = 0$ implies $A^{\rot}y_1 = 0$. Since $A$ has full row rank, $y_1 = 0$. Thus the dual contains no line and the minimum is attained at a vertex of its feasible region. The feasible region of the dual does not depend on $x$. 

	Let $(y^1,\zp_1,\zn_1)$ to $(y^K,\zp_K,\zn_K)$ be the vertices of~\eqref{dual LP}, and let $\Y_A = \sset{\abs{A^{\rot} y^1},\ldots,\abs{A^{\rot} y^K}}$. Then
\begin{align*}
\min_{d \in \Y_{A}} d^{\rot} x &= \min\{x^{\rot} (\zp + \zn) : - b^{\rot} y = 1, A^{\rot}y + \zp - \zn = 0,\  \zp,\zn\ge 0\} \\
&= \min \{\abs{y^{\rot} A} x: b^{\rot}y = -1\}.
\end{align*}

We finally show that there is a feasible $f$ with $\abs{f} \le x/\Copt$. Let $x' := x/\Copt$.
Then $x' > 0$ and $\min_{d \in \Y_{A}} d^{\rot} x' = \min_{d \in \Y_{A}} d^{\rot} x/\Copt =
 \Copt/\Copt = 1$ and thus the right LP with $x = x'$ \eqref{LP-Duality} has objective value 1. Hence, the left LP has objective value 1 and there is a feasible $f$ with $\abs{f} \le x'$.
\end{proof}

\subsection{Convergence to Dominance}
In the network setting, an important role is played by the set of edge capacity vectors that support a feasible flow. In the LP setting, we generalize this notion to the set of \emph{dominating states}, which is defined as 
\[
    \Xdom := \{x \in \R^m: 
    \exists \text{ feasible } f:\abs{f} \le x\}. 
\]
An alternative characterization, using the set $\Y_{A}$ from Lemma \ref{dual}, is
\[
    \XX_1 := \{x \in \R_{\ge 0}^m:d^{\rot} x \ge 1 \text{ for all } d \in \Y_{A}\}.
\]
We now prove that $\Xdom=\XX_1$ and that the set $\XX_1$ is an attractor in the following sense.
%We now prove that $\Xdom=\XX_1$ and the set $\XX_1$ is attracting in the sense that the distance between $x(t)$ and $\XX_1$ goes to zero, as $t$ increases.
\begin{lemma}\label{dominating states}
	The following statements hold:
	\begin{enumerate}
		\item $\Xdom=\XX_1$. Moreover, $\lim_{t \rightarrow \infty} \dist(x(t),\XX_1) = 0$, where $\dist(x,\XX_1)$ is the Euclidean distance between $x$ and $\XX_1$.
		\item If $x(t_0) \in \XX_1$, then $x(t) \in \XX_1$ for all $t \ge t_0$. For all sufficiently large $t$, $x(t) \in \XX_{1/2} := \{x \in \R_{\ge 0}^n:d^{\rot} x \ge 1/2 \text{ for all } d \in \Y_{A}\}$, and if $x \in \XX_{1/2}$ then there is a feasible $f$ with $\abs{f} \le 2x$.
	\end{enumerate}
\end{lemma}
\begin{proof} 
	\begin{enumerate}
		\item If $x \in \XX_1$, then $d^{\rot} x \ge 1$ for all $d \in \Y_{A}$ and hence Lemma~\ref{dual} implies the existence of a feasible solution $f$ with $\abs{f} \le x$. Conversely, if $x \in \Xdom$, then there is a feasible $f$ with $\abs{f} \le x$. Thus $d^{\rot} x \ge 1$ for all $d \in \Y_{A}$ and hence $x \in \XX_1$.
 		By the proof of Lemma~\ref{dual}, for any $d \in \Y_{A}$, there is a $y$ such that $d = \abs{A^{\rot} y}$ and $b^{\rot}y = -1$.
	Let $Y(t) = d^{\rot} x$. Then
	\[      \dot{Y} = \abs{y^{\rot} A} \dot{x} = \abs{y^{\rot} A}(\abs{q} - x) \ge \abs{y^{\rot} A q} - Y = \abs{y^{\rot} b} - Y = 1 - Y.\]
	Thus for any $t_0$ and $t \ge t_0$, $Y(t) \ge 1 + (Y(t_0) - 1)e^{-(t - t_0)}$ 
	by Lemma~\ref{Gronwalls Lemma} applied with $A = 1$ and $\alpha = -1$. 
	In particular,  $\lim \inf_{t \rightarrow \infty} Y(t) \ge 1$. Thus $\lim \inf_{t \rightarrow \infty} \min_{d \in \Y_{A}} d^{\rot} x \ge 1$ and hence $\lim_{t \rightarrow \infty} \dist(x(t),\XX_1) = 0$. 
	\item Moreover, if $Y(t_0) \ge 1$, then $Y(t) \ge 1$ for all $t \ge t_0$. Hence $x(t_0) \in \XX_1$ implies $x(t) \in \XX_1$ for all $t \ge t_0$. Since $x(t)$ converges to $\XX_1$, $x(t) \in \XX_{1/2}$ for all sufficiently large $t$.
	If $x \in \XX_{1/2}$ there is $f$ such that $Af = \tfrac{1}{2} b$ and $\abs{f} \le x$. Thus $2f$ is feasible and $\abs{2f} \le 2x$.\qedhere
\end{enumerate}
\end{proof}

The next lemma summarizes simple bounds on the values of resistors $r$, 
potentials $p$ and states $x$ that hold for sufficiently large $t$. 
Recall that	$P = \set{e\in[m]}{c_e > 0}$ and $Z=\set{e\in[m]}{c_e = 0}$, see~\eqref{eq:idxSetPZ}.

\begin{lemma}\label{Simple Facts}
	The following statements hold:
    \begin{enumerate}
		\item For sufficiently large $t$, it holds that
		      $r_e \ge c_e/ (2\D\onenorm{b/\gamma_{A}})$,
                $b^{\rot} p \le  8 \D\onenorm{b/\gamma_{A}} \onenorm{c}$
		      	and $\abs{A_e^{\rot} p} \le 8 \D^2\onenorm{b} \onenorm{c}\text{ for all e}$.
        \item
		      For all $e$, it holds that $\dot{x}_e / x_e \ge -1$ and for all $e\in P$,
              it holds that $\dot{x}_e / x_e \le 8\D^2\onenorm{b}\onenorm{c}/\cmin$.
		\item There is a positive constant $C$ such that for all $t \ge t_0$, there is a
                feasible $f$ (depending on $t$) such that  $x_e(t) \ge C$ for all indices
                $e$ in the support of $f$.
	\end{enumerate}
\end{lemma}
\begin{proof}
	\begin{enumerate}
		\item By Lemma~\ref{existence}, $x_e(t) \le 2\D\onenorm{b/\gamma_{A}}$ for
		 all sufficiently large $t$. It follows that
		 $r_e =c_e/x_e \ge c_e/ (2\D\onenorm{b/\gamma_{A}})$. 
		 Due to Lemma~\ref{dominating states}, for large enough $t$, 
		 there is a feasible flow with $|f| \le 2x$. 
		 Together with $x_e(t) \le 2\D\onenorm{b/\gamma_{A}}$, it follows that 
		 \[
		 b^{\rot} p = E(q) \le E(2x) = 4 c^{\rot} x \le 8 \D\onenorm{b/\gamma_{A}} \onenorm{c}.
		 \]
		Now, orient $A$ according to $q$ and consider any index $e'$. Recall that for all 
		indices $e$, we have $A_e^{\rot} p = 0$ if $e\in Z$, and 
		$q_e = (x_e/c_e) \cdot A_e^{\rot} p $ if $e\in P$. 
		Thus $A_e^{\rot} p \ge 0$ for all $e$.
		If $e' \in Z$ or $e' \in P$ and $q_{e'} = 0$, the claim is obvious. So assume $e' \in P$ and $q_{e'} > 0$. Since $q$ is a convex combination of $q$-sign-compatible basic feasible solutions, there is a basic feasible solution $f$ with $f \ge 0$ and $f_{e'} > 0$. By Lemma~\ref{lem:fsDg},
		$f_{e'} \ge 1/(\D\gamma_{A})$. Therefore
		\[ 
		f_{e'} A_{e'}^{\rot} p \le \sum_e f_e A_e^{\rot} p = b^{\rot} p \le  8 \D\onenorm{b/\gamma_{A}} \onenorm{c}
		\]
		for all sufficiently large $t$. The inequality follows from $f_e \ge 0$ and $A_e^{\rot} p \ge 0$ for all $e$. Thus $A_{e'}^{\rot} p \le 8 \D^2\onenorm{b} \onenorm{c}$ for all sufficiently large $t$.
		\item We have $\dot{x}_e/x_e = (\abs{q_e} - x_e)/x_e \ge - 1$ for all $e$. For $e$ with $c_e > 0$
		\[
			\frac{\dot{x}_e}{x_e} = \frac{\abs{q_e} - x_e}{x_e} \le 
			\frac{\abs{A_e^{\rot} p}}{c_e} \le 8\D^2\onenorm{b}\onenorm{c}/\cmin.
		\]
		\item Let $t_0$ be such that $d^{\rot} x(t) \ge 1/2$ for all $d \in \Y_{A}$ and $t \ge t_0$. Then for all $t \ge t_0$, there is $f$ such that $Af = \frac{1}{2} b$ and $\abs{f} \le x(t)$; $f$ may depend on $t$. By Lemma~\ref{sign-compatible representation}, we can write $2f$ as convex combination of $f$-sign-compatible basic feasible solutions (at most $m$ of them) and a $f$-sign-compatible solution in the kernel of $A$. Dropping the solution in the kernel of $A$ leaves us with a solution which is still dominated by $x$.

		It holds that for every $e\in E$ with $f_e\neq 0$, there is a basic feasible solution $g$ 
		used in the convex decomposition such that $2\abs{f_e} \ge \abs{g_e}>0$. 
		By Lemma~\ref{lem:fsDg}, every non-zero component of $g$ is at least $1/(\D\gamma_{A})$. 
		We conclude that $x_e \ge 1/(2\D\gamma_{A})$, for every $e$ in the support of $g$.\qedhere
	\end{enumerate}
\end{proof}

\subsection{The Equilibrium Points}

We next characterize the equilibrium points 
\begin{equation}\label{eq:defF}
F = \set{x \in \R_{\ge 0}}{\abs{q} = x}.
\end{equation}
Let us first elaborate on the special case of the undirected shortest path problem. Here the equilibria are the flows of value one from source to sink in a network formed by undirected source-sink paths of the same length. This can be seen as follows. Consider any $x \ge 0$ and assume $\supp(x)$ is a network of undirected source-sink paths of the same length. Call this network $\mathcal{N}$. Assign to each node $u$, a potential $p_u$ equal to the length of the shortest undirected path from the sink $s_1$ to $u$. These potentials are well-defined as all paths from $s_1$ to $u$ in $\mathcal{N}$ must have the same length. For an edge $e= (u,v)$ in $\mathcal{N}$, we have $q_e = x_e/c_e(p_u - p_v) = x_e/c_e \cdot c_e = x_e$, i.e., $q = x$ is the electrical flow with respect to the resistances $c_e/x_e$.
Conversely, if $x$ is an equilibrium point and the network is oriented such that $q \ge 0$, we have $x_e = q_e = x_e/c_e (p_u - p_v)$ for all edges $e = (u,v) \in \supp(x)$. Thus $c_e = p_u - p_v$ and this is only possible if for every node $u$, all paths from $u$ to the sink have the same length. Thus $\supp(x)$ must be a network of undirected source-sink paths of the same length.
We next generalize this reasoning.

\begin{theorem}\label{equilibrium}
	If $x = \abs{q}$ is an equilibrium point and the columns of $A$ are oriented such that $q \ge 0$, then all feasible solutions $f$ with $\supp(f) \subseteq \supp(x)$ satisfy $c^{\rot} f = c^{\rot} x$. Conversely, if $x = \abs{q}$ for a feasible $q$, $A$ is oriented such that $q \ge 0$, and all feasible solutions $f$ with
	$\supp(f) \subseteq \supp(x)$ satisfy $c^{\rot} f = c^{\rot} x$, then $x$ is an equilibrium point.
\end{theorem}

\begin{proof}
	If $x$ is an equilibrium point, $\abs{q_e} = x_e$ for every $e$. 
	By changing the signs of some columns of $A$, we may assume $q \ge 0$, i.e., $q = x$. 
	Let $p$ be the potential with respect to $x$. For every index $e\in P$ in the support of $x$,
	since $c_e>0$ we have $q_e = \tfrac{x_e}{c_e} A_e^{\rot} p$ and hence $c_e = A_e^{\rot} p$. 
	Further, for the indices $e\in Z$ in the support of $x$, we have $c_e = 0 = A_e^{\rot} p$ 
	due to the second block of equations on the right hand side 
	in~\eqref{both blocks}. 
	Let $f$ be any feasible solution whose support is contained in the support of $x$. 
	Then the first part follows by
	\begin{equation*}
		\sum_{e \in \supp(f)} c_e f_e = \sum_{e \in \supp(f)} f_e A_e^{\rot} p = b^{\rot} p = E(q) = E(x) = \cost(x).
	\end{equation*}
	For the second part, we misuse notation and use $A$ to also denote the submatrix of the constraint matrix indexed by the columns in the support of $x$. We may assume that the rows of $A$ are independent. Otherwise, we simply drop redundant constraints. We may assume $q \ge 0$; otherwise we simply change the sign of some columns of $A$. Then $x$ is feasible. Let $A_B$ be a square non-singular submatrix of $A$ and let $A_N$ consist of the remaining columns of $A$. The feasible solutions $f$ with $\supp(f) \subseteq \supp(x)$ satisfy $A_B f_B + A_N f_N = b$ and hence $f_B = A_B^{-1}(b - A_N f_N)$.
	Then
    \[
	c^{\rot} f = c_B^{\rot} f_B + c_N^{\rot} f_N = c_B A_B^{-1} b + (c_N^{\rot} - c_B^{\rot} A_B^{-1} A_N) f_N.
    \]
	
Since, by assumption, $c^{\rot} f$ is constant for all feasible solutions whose support is contained in the support of $x$, we must have
	$c_N = A_N^{\rot} [A_B^{-1}]^{\rot} c_B$. Let $p = [A_B^{-1}]^{\rot} c_B$. Then it follows that
	$A^{\rot} p = \big[ \begin{smallmatrix} A_B^{\rot}  \\ A_N^{\rot} \end{smallmatrix} \big] 
	[A_B^{-1}]^{\rot} c_B = \big[ \begin{smallmatrix} c_B  \\ c_N  \end{smallmatrix} \big]$
	and hence $R x = A^{\rot} p$. Thus the pair $(x,p)$ satisfies the right hand side of~\eqref{both blocks}. Since $x$ is feasible, it also satisfies the left hand side of~\eqref{both blocks}. Therefore, $x$ is the minimum energy solution with respect to $x$.
\end{proof}

\begin{corollary}
	Let $g$ be a basic feasible solution. Then $\abs{g}$ is an equilibrium point.
\end{corollary}
\begin{proof}
	Let $g$ be a basic feasible solution. Orient $A$ such that $g \ge 0$. Since $g$ is basic, there is a $B \subseteq [m]$ such that $g = (g_B,g_N) = (A_B^{-1}b,0)$. Consider any feasible solution $f$ with $\supp(f) \subseteq \supp(g)$. Then $f = (f_B,0)$ and hence $b = Af = A_B f_B$. Therefore, $f_B = g_B$ and hence $c^{\rot} f = c^{\rot} g$. Thus $x = \abs{g}$ is an equilibrium point.
\end{proof}
This characterization of equilibria has an interesting consequence.
\begin{lemma}
	The set $L:=\{c^{\rot} x:x\in F\}$ of costs of equilibria is finite.
\end{lemma}
\begin{proof}
	If $x$ is an equilibrium, $x = \abs{q}$, where $q$ is the minimum energy solution with respect to $x$. Orient $A$ such that $q \ge 0$. Then by Theorem~\ref{equilibrium}, $c^{\rot} f = c^{\rot} x$ for all feasible solutions $f$ with $\supp(f) \subseteq \supp(x)$. In particular, this holds true for all such basic feasible solutions $f$. Thus $L$ is a subset of the set of costs of all basic feasible solutions, which is a finite set.
\end{proof}
We conclude this part by showing that the optimal solutions
of the undirected linear program~\eqref{ULP} are equilibria.

\begin{theorem}\label{optima are equilibria}
	Let $x$ be an optimal solution to~\eqref{ULP}. Then $x$ is an equilibrium.
\end{theorem}
\begin{proof} By definition, there is a feasible $f$ with $\abs{f} = x$. Let us reorient the columns of $A$ such that $f \ge 0$ and let us delete all columns $e$ of $A$ with $f_e = 0$. Consider any feasible $g$ with $\supp(g) \subseteq \supp(x)$. We claim that $c^{\rot} x = c^{\rot} g$. Assume otherwise and consider the point $y = x + \lambda(g - x)$. If $\abs{\lambda}$ is sufficiently small, $y \ge 0$. Furthermore, $y$ is feasible and $c^{\rot} y = c^{\rot} x + \lambda (c^{\rot} g - c^{\rot} x)$. If $c^{\rot} g \not= c^{\rot} x$, $x$ is not an optimal solution to~\eqref{ULP}. The claim now follows from Theorem~\ref{equilibrium}.
\end{proof}

\subsection{Convergence}

In order to show convergence, we construct a Lyapunov function. The following functions play a crucial role in our analysis.
Let $C_d = d^{\rot} x$ for $d \in \Y_{A}$, and recall that $\Copt = \min_{d \in \Y_{A}}d^{\rot} x$ denotes the optimum. Moreover, we define
\begin{align*}
	h(t) :=   \sum_e r_e \abs{q_e} \frac{x_e}{\Copt} - E\left(\frac{x}{\Copt}\right)
	\quad \text{ and }\quad
	V_d := \frac{c^{\rot} x}{C_d} \text{ for every $d \in \Y_{A}$}
	.
\end{align*}

\begin{theorem} \label{Lyapunov functions}
	\begin{compactenum}[\quad(1)]
		\item For every $d \in \Y_{A}$, $\dot{C}_d\ge 1-C_d$. Thus, if $C_d  < 1$ then $\dot{C_d} > 0$.
		\item If $x(t) \in \XX_1$, then $\frac{d}{dt} \cost(x(t)) \le 0$ with equality if and only if $x = \abs{q}$.
		\item Let $d \in \Y_{A}$ be such that $\Copt = d^{\rot} x$ at time $t$. Then it holds that $\dot{V}_d \le h(t)$.
		\item It holds that $h(t) \le 0$ with equality if and only if  $\abs{q} = \tfrac{x}{\Copt}$.
	\end{compactenum}
\end{theorem}
\begin{proof}
	\begin{compactenum}
		\item Recall that for $d \in \Y_{A}$, there is a $y$ such that $b^{\rot} y = -1$ and $d = \abs{A^{\rot} y}$. Thus
		$\dot{C_d} = d^{\rot}  (\abs{q} - x) \ge \abs{y^{\rot} A q} - C_d = 1 - C_d$ and hence $\dot{C_d} > 0$, whenever $C_d < 1$.
		\item\label{part2} Remember that $E(x) = \cost(x)$ and that $x(t) \in \XX_1$ implies that there is a feasible $f$ with $\abs{f} = x$. Thus
		$E(q) \le E(f) \le E(x)$. Let $R$ be the diagonal matrix of entries $c_e/x_e$. Then
		\begin{align*}
			\frac{d}{dt}{\cost(x)} & = c^{\rot} (\abs{q} - x)                           & \text{by~\eqref{PD}}                               \\
			& = x^{\rot} R^{1/2} R^{1/2} \abs {q} - x^{\rot} R x & \text{since $c = Rx$}\\
			                       & \le (q^{\rot} R q)^{1/2} (x^{\rot} R x)^{1/2} - x^{\rot} R x & \text{by Cauchy-Schwarz}                           \\
			                       & \le 0 & \text{since $E(q) \le E(x)$.}
		\end{align*}
		If the derivative is zero, both inequalities above have to be equalities. This is only possible if the vectors $\abs{q}$ and $x$ are parallel and $E(q) = E(x)$. Let $\lambda$ be such that $\abs{q} = \lambda x$. Then
			$E(q) = \sum_e \tfrac{c_e}{x_e} q_e^2 = \lambda^2 \sum_e c_e x_e = \lambda^2 E(x).$
		Since $E(x) > 0$, this implies $\lambda = 1$.
		\item By definition of $d$, $\Copt = C_d$. By the first two items, we have
		$\dot{\Copt} = d^{\rot} \abs{q} - \Copt$  and $\frac{d}{dt} \cost(x) = c^{\rot} \abs{q} - \cost(x)$.
		Thus
		\begin{align*}
			\frac{d}{dt} \frac{\cost(x)}{\Copt} 
			& = \frac{ \Copt \frac{d}{dt}{\cost(x)} - \dot{\Copt} \cost(x) }{\Copt^2}            
			=  \frac{ \Copt (c^{\rot} \abs{q} - \cost(x)) - (d^{\rot} \abs{q} - \Copt)\cost(x)}{\Copt^2} \\
			& = \frac{\Copt \cdot c^{\rot} \abs{q} - d^{\rot} \abs{q} \cdot c^{\rot} x}{\Copt^2}          
			\le \sum_e r_e \abs{q_e} \frac{x_e}{\Copt} - \sum_e r_e \Big(\frac{x_e}{\Copt}\Big)^2 
			= h(t),
		\end{align*}
		where we used $r_e = c_e/x_e$ and hence $c^{\rot} \abs{q} = \sum_e r_e x_e \abs{q_e}$, $c^{\rot}x = E(x)$, and $d^{\rot} \abs{q} \geq \abs{y^{\rot} A q} = 1$ since $d = \abs{y^{\rot} A}$ for some $y$ with $b^{\rot} y = -1$.
		\item We have
		\[ 
			\sum_e r_e \tfrac{x_e}{\Copt}  \abs{q_e} 
			= \sum_e r_e^{1/2} \tfrac{x_e}{\Copt}  r_e^{1/2} \abs{q_e} 
			\le \Big( \sum_e r_e (\tfrac{x_e}{\Copt})^2 \Big)^{1/2}
			\Big( \sum_e r_e q_e^2 \Big)^{1/2} 
			=\E\big(\tfrac{x}{\Copt}\big)^{1/2} \E(q)^{1/2}
		\]
		by Cauchy-Schwarz. Since $h(t)=\sum_e r_e \abs{q_e} \tfrac{x_e}{\Copt} - E(\tfrac{x}{\Copt})$ by definition, it follows that
		\[  
			h(t) 
			\le \E\big(\tfrac{x}{\Copt}\big)^{1/2} \cdot \E(q)^{1/2} - \E\big(\tfrac{x}{\Copt}\big)
			= E\big(\tfrac{x}{\Copt}\big)^{1/2}\cdot \Big( \E(q)^{1/2} - \E\big(\tfrac{x}{\Copt}\big)^{1/2}\Big) \le 0
		\]
		since $x/\Copt$ dominates a feasible solution and hence $\E(q) \le \E(x/\Copt)$.
		If $h(t) = 0$, we must have equality in the application of Cauchy-Schwarz, i.e., the vectors $x/\Copt$ and $\abs{q}$ must be parallel, and we must have $E(q) = E(x/\Copt)$ as in the proof of part~\ref{part2}. \qedhere
	\end{compactenum}
\end{proof}

We show now convergence against the set of equilibrium points.
We need the following technical Lemma from~\cite{Physarum}.

\begin{lemma}[Lemma 9 in~\cite{Physarum}]\label{f=fd}
	Let $f(t) = \max_{d \in \Y_{A}} f_d(t)$, where each $f_d$ is continuous and differentiable. If $\dot{f}(t)$ exists, then there is a $d \in \Y_{A}$ such that $f(t) = f_d(t)$ and $\dot{f}(t) = \dot{f}_d(t)$.
\end{lemma}

\begin{theorem}
	All trajectories converge to the set $F$ of equilibrium points.
\end{theorem}
\begin{proof}
	We distinguish cases according to whether the trajectory ever enters $\XX_1$ or not.
	If the trajectory enters $\XX_1$, say $x(t_0) \in \XX_1$, then $\frac{d}{dt} \cost(x) \le 0$ for all $t \ge t_0$ with equality only if $x = \abs{q}$. Thus the trajectory converges to the set of fix points.
	If the trajectory never enters $\XX_1$, consider $V = \max_{d \in \Y_{A}} (V_d + 1 - C_d)$. We show that $\dot{V}$ exists for almost all $t$. Moreover, if $\dot{V}(t)$ exists, then $\dot{V}(t) \le 0$ with equality if and only if $\abs{q_e} = x_e$ for all $e$. It holds that $V$ is Lipschitz continuous as the maximum of a finite number of continuously differentiable functions. Since $V$ is Lipschitz continuous, the set of $t$'s where $\dot{V}(t)$ does not exist has zero Lebesgue measure (see for example~\cite[Ch.~3]{Clarke:1998}).
	If $\dot{V}(t)$ exists, we have $\dot{V}(t) = \dot{V}_d(t) - \dot{C}_d(t)$ for some $d \in \Y_{A}$ according to Lemma~\ref{f=fd}. Then, it holds that  $\dot{V}(t) \le h(t) - (1 - C_d) \le 0$.
	Thus $x(t)$ converges to the set
	\begin{align*}
		\set{x \in \R_{\ge 0}}{ \dot{V} = 0} & = \set{x \in \R_{\ge 0}}{\text{$\abs{q} = x/C$ and $C = 1$}}
		= \set{x \in \R_{\ge 0}}{\abs{q} = x}.\qedhere
	\end{align*}
\end{proof}

At this point, we know that all trajectories $x(t)$ converge to $F$. 
Our next goal is to show that $c^{\rot} x(t)$ converges to the cost of 
an optimum solution of~\eqref{ULP} and that $\abs{q} - x$ converges to zero. 
We are only able to show the latter for all indices $e\in P$, i.e. with $c_e>0$.

\subsection{Details of the Convergence Process}

In the argument to follow, we will encounter the following situation several times. We have a non-negative function $f(t) \ge 0$ and we know that $\int_0^\infty f(t)dt$ is finite. We want to conclude that $f(t)$ converges to zero for $t \rightarrow \infty$. This holds true if $f$ is Lipschitz continuous. Note that the proof of the following lemma is very similar to the proof in~\cite[Lemma 11]{Physarum}. However, in our case we apply the local Lipschitz condition that we showed in Lemma~\ref{Locally Lipschitz, General Case}.

\begin{lemma}\label{Lipschitz}
	Let $f(t) \ge 0$ for all $t$. If $\int_0^\infty f(t) d(t)$ is finite and $f(t)$ is locally Lipschitz continuous, i.e., for every $\varepsilon > 0$, there is a $\delta >0$ such that $\abs{f(t') - f(t)} \le \varepsilon$ for all $t' \in [t,t+\delta]$, then $f(t)$ converges to zero as $t$ goes to infinity. The functions $t \mapsto x^{\rot} R \abs{q} - x^{\rot} R x = c^{\rot} \abs{q} - c^{\rot} x$ and $t \mapsto h(t)$ are Lipschitz continuous.
\end{lemma}
\begin{proof}
	If $f(t)$ does not converge to zero, there is $\varepsilon> 0$ and an infinite unbounded sequence $t_1$, $t_2$, \ldots  such that $f(t_i) \ge \varepsilon$ for all $i$. Since $f$ is Lipschitz continuous there is $\delta > 0$ such that $f(t_i') \ge \varepsilon/2$ for $t'_i \in [t_i,t_i + \delta]$ and all $i$. Hence, the integral $\int_0^\infty f(t) dt$ is unbounded.

	Since $\dot{x}_{e}$ is continuous and bounded (by Lemma~\ref{x is bounded}), $x_e$ is
	Lipschitz continuous. Thus, it is enough to show that $q_e$ is
	Lipschitz continuous for all $e$.  Since $q_Z$ (recall that $Z = \set{e}{c_e = 0}$ and $P = [m] \setminus Z$) is an affine function of $q_P$, it suffices to establish the claim for $e \in P$. 
	So let $e\in P$ be such that $c_e>0$.
	First, we claim that $x_e(t+\varepsilon) \le (1+2K \varepsilon) x_e$ for all
	$\varepsilon \le K/4$, where $K=8\D^2\onenorm{b}\onenorm{c}/\cmin$.
	Assume that this is not the case. Let $$\varepsilon=\inf \{ \delta \le 1/4K : x_e(t+\delta) > (1+2K\delta) x_e(t) \},$$
	then $\varepsilon > 0$ (since $\dot{x}_{e}(t) \le K x_e(t)$ by Lemma
	\ref{Simple Facts}) and, by continuity, $x_e(t+\varepsilon) \ge
	(1+2K\varepsilon) x_e(t)$.
	There must be $t' \in [t,t+\varepsilon]$ such that $\dot{x}_{e}(t') = 2K x_e(t)$. On the other hand,
	\begin{align*}
		\dot{x}_{e}(t') \le K x_e(t') \le K (1+2K \varepsilon) x_e(t)  \le K (1+2K/4K) x_e(t) < 2 K x_e(t),
	\end{align*}
	which is a contradiction. Thus, $x_e(t+\varepsilon) \le (1+2K \varepsilon) x_e$ for all $\varepsilon \le 1/4K$. Similarly, $x_e(t+\varepsilon) \ge (1-2K \varepsilon) x_e$.
	Now, let $\alpha = (1 - 2K \varepsilon) x_e$ and $\beta = (1 + 2K \varepsilon) x_e$. Then
	\[ 
	\abs{\abs{q_e(t+\delta)} - \abs{q_e(t)}} \le M \onenorm{x(t + \delta) - x(t)} 
	\le M m (4K \varepsilon) x_e \le 8\varepsilon M m K \D \onenorm{b/\gamma_{A}}, 
	\]
	since $x_e \le 2\D\onenorm{b/\gamma_{A}}$ for sufficiently large $t$ and where $M$ is as in Lemma~\ref{Locally Lipschitz, General Case}.
	Since $\Copt$ is at least $1/2$ for all sufficiently large $t$, the division by $\Copt$ and $\Copt^2$ in the definition of $h(t)$ does not affect the claim.
\end{proof}

\begin{lemma}\label{xgoestoq}
	For all $e \in [m]$ of positive cost, it holds that $\abs{x_e - \abs{q_e}} \rightarrow 0$ 
	as $t$ goes to infinity.
\end{lemma}
\begin{proof}
	For a trajectory ultimately running in $\XX_1$, we showed $\frac{d}{dt} \cost(x) \le x^{\rot} R \abs{q} - x^{\rot} R x \le 0$ with equality if and only if $x = \abs{q}$. Also, $E(q) \le E(x)$, since $x$ dominates a feasible solution. Furthermore, $x^{\rot} R \abs{q} - x^{\rot} R x$ goes to zero using Lemma~\ref{Lipschitz}. Thus
	\[
		\sum_e r_e (x_e - \abs{q_e})^2 
		= \sum_e r_e x_e^2 + \sum_e r_e q_e^2 - 2 \sum_e r_e x_e \abs{q_e} 
		\le 2\Big(\sum_e r_e x_e^2 - \sum_e r_e x_e \abs{q_e}\Big)
	\]
	goes to zero.
	Next observe that there is a constant $C$ such that $x_e(t) \le C$ for all $e$ and $t$ as a result of Lemma~\ref{x is bounded}.
	Also $\cmin > 0$ and hence $r_e \ge \cmin/C$. Thus $\sum_e r_e (x_e - \abs{q_e})^2 \le \tfrac{C}{\cmin} \cdot \sum_e  (x_e - \abs{q_e})^2$ and hence $\abs{x_e - \abs{q_e}} \rightarrow 0$ for every $e$ with positive cost.
	For trajectories outside $\XX_1$, we argue about $|\abs{q_e} - \frac{x}{\Copt}|$ and use $\Copt \rightarrow 1$, namely
	\begin{align*}
		\sum_e r_e (\tfrac{x_e}{\Copt} - \abs{q_e})^2
		  \le 2 \Big(\sum_e r_e (\tfrac{x_e}{\Copt})^2 - \sum_e r_e \tfrac{x_e}{\Copt} \abs{q_e}\Big)  \rightarrow 0.&\qedhere
	\end{align*}
\end{proof}

Note that the above does not say anything about the indices $e\in Z$ (with $c_e=0$). Recall that $A_P q_P + A_Z q_Z = b$ and that the columns of $A_Z$ are independent. Thus, $q_Z$ is uniquely determined by $q_P$.
For the undirected shortest path problem, the potential difference $p^{\rot} b$ between source and sink converges to the length of  a shortest source-sink path. If an edge with positive cost is used by some shortest undirected path, then no shortest undirected path uses it with the opposite direction. We prove the natural generalizations.

Let $\OPT$ be the set of optimal solutions to \eqref{ULP} and let 
$\Eopt = \cup_{x \in \OPT} \supp(x)$ be the set of columns used in some optimal solution. 
The columns of positive cost in $\Eopt$ can be consistently oriented as the following Lemma shows.

\begin{lemma}
	Let $x_1^*$ and $x_2^*$ be optimal solutions to \eqref{ULP} and let $f$ and $g$ be feasible solutions with $\abs{f} = x^*_1$ and $\abs{g} = x^*_2$. Then there is no $e$ such that $f_e g_e < 0$ and $c_e > 0$.
\end{lemma}
\begin{proof}
	Assume otherwise. Then $\abs{g_e - f_e} = \abs{g_e} + \abs{f_e} > 0$. Consider $h = (g_e f  - f_e g)/(g_e - f_e)$. Then $Ah = (g_e Af - f_e Ag)/(g_e - f_e) = b$ and $h$ is feasible.
	Also, $h_e = \tfrac{g_e f_e - f_e g_e}{g_e - f_e} = 0$ and for every index $e'$, it holds that $\abs{h_{e'}} =\tfrac{\abs{g_e f_{e'} - f_e g_{e'}}}{\abs{g_e - f_e)}} \le \tfrac{\abs{g_e} \abs{f_{e'}} + \abs{f_e}\abs{ g_{e'}}}{\abs{g_e} + \abs{f_e}}$
	and hence
	\[ \cost(h) < \cost(f) + \cost(g) = \frac{|g_e|}{|g_e|+|f_e|}\cost(x_1^*) + \frac{|f_e|}{|g_e|+|f_e|}\cost(x_2^*) = \cost(x_1^*), \]
	a contradiction to the optimality of $x_1^*$ and $x_2^*$.
\end{proof}
By the preceding lemma, we can orient $A$ such that $f_e \ge 0$ whenever $\abs{f}$ is an optimal solution to~\eqref{ULP} and $c_e > 0$. We then call $A$ \emph{positively oriented}. 
\begin{lemma}\label{ptb->ctx*}
	It holds that $p^{\rot} b$ converges to the cost of an optimum solution of~\eqref{ULP}. If $A$ is positively oriented, then $\liminf_{t \rightarrow \infty} A_e^{\rot} p \ge 0$ for all $e$.
\end{lemma}
\begin{proof}
	Let $x^*$ be an optimal solution of~\eqref{ULP}. We first show convergence to a point in $L$ and then convergence to $c^{\rot} x^*$.
	Let $\varepsilon > 0$ be arbitrary. Consider any time $t \ge t_0$, where $t_0$ and $C$ as in Lemma~\ref{Simple Facts} and moreover $\abs{\abs{q_e} - x_e} \le \tfrac{C\varepsilon}{\cmax}$ for every $e\in P$. Then $x_e \ge C$ for all indices $e$ in the support of some basic feasible solution $f$. 
	For every $e\in P$, we have $q_e = \tfrac{x_e}{c_e} A_e^{\rot} p$. We also assume $q \ge 0$ by possibly reorienting columns of $A$. Hence
	\[ 
		\left|c_{e}-A_{e}^{\rot}p\right|=\left|1-\frac{q_{e}}{x_{e}}\right|\cdot\big|c_{e}=\left|\frac{x_{e}-q_{e}}{x_{e}}\right|\cdot c_{e}\le\frac{\cmax}{C}\left|q_{e}-x_{e}\right|\le\varepsilon.
	\]
	For indices $e\in Z$, we have $A_e^{\rot}p = 0 = c_e$.
	Since $\infnorm{f}\leq \D \onenorm{b/\gamma_{A}}$ (Lemma~\ref{lem:fsDg}), we conclude
	\[   
	c^{\rot} f - p^{\rot} b = \sum_{e \in \supp(f)} (c_e - p^{\rot} A_e) f_e \le \epsilon \sum_e \abs{f_e} 
	\le \epsilon \cdot m \D \onenorm{b/\gamma_{A}}.
	\]
	Since the set $L$ is finite, we can let $\epsilon > 0$ be smaller than half the minimal distance between elements in $L$. By the preceding paragraph, there is for all sufficiently large $t$, a basic feasible solution $f$ such that $\abs{c^{\rot} f - b^{\rot} p} \le \epsilon$. Since $b^{\rot} p$ is a continuous function of time, $c^{\rot} f$ must become constant. We have now shown that $b^{\rot} p$ converges to an element in $L$.
	We will next show that $b^{\rot}p$ converges to the optimum cost. Let $x^*$ be an optimum solution to~\eqref{ULP} and let $W = \sum_e x_e^* c_e \ln x_e$. Since $x(t)$ is bounded, $W$ is bounded. We assume that $A$ is positively oriented, thus there is a feasible $f^*$ with $\abs{f^*} = x^*$ and $f_e^* \ge 0$ whenever $c_e > 0$. By reorienting zero cost columns, we may assume $f^*_e \ge 0$ for all $e$. Then $A x^* = b$. We have
	\begin{align*}
		\dot{W} &= \sum_e x^*_e c_e \frac{\abs{q_e} - x_e}{x_e} \\
		  & = \sum_{e;\ c_e > 0} x^*_e  \left|A_{e}^{\rot}p\right| - \cost(x^*) & \text{since $q_e = \tfrac{x_e}{c_e} A_e^{\rot} p$ whenever $c_e > 0$} \\
		  & = \sum_e x^*_e \left|A_{e}^{\rot}p\right| - \cost(x^*)             & \text{since $A_e^{\rot} p = 0$ whenever $c_e = 0$}             \\
		&= \sum_e x^*_e \Big(\left|A_{e}^{\rot}p\right| - A_e^{\rot} p\Big) + b^{\rot} p - \cost(x^*)
	\end{align*}
	and hence $b^{\rot} p - \cost(x^*)$ must converge to zero; note that $b^{\rot}p$ is Lipschitz continuous in $t$.

	Similarly,  $\abs{A_e^{\rot} p} - A_e^{\rot} p$ must converge to zero whenever $x_e^* > 0$. This implies $\liminf A_e^{\rot} p \ge 0$. Assume otherwise, i.e., for every $\varepsilon > 0$, we have $A_e^{\rot}p < -\varepsilon$ for arbitrarily large $t$. Since $p$ is Lipschitz continuous in $t$, there is a $\delta > 0$ such that $A_e^{\rot}p <  -\varepsilon/2$ for infinitely many disjoint intervals of length $\delta$. In these intervals,  $\abs{A_e^{\rot} p} - A_e^{\rot} p \ge \varepsilon$ and hence
	$W$ must grow beyond any bound, a contradiction.
\end{proof}

\begin{corollary}
	$E(x)$ and $\cost(x)$ converge to $c^{\rot} x^*$, whereas $x$ and $|q|$ converge to $\OPT$.
    If the optimum solution is unique, $x$ and $|q|$ converge to it. Moreover,
    if $e\notin \Eopt$, $x_e$ and $\abs{q_e}$ converge to zero.
\end{corollary}
\begin{proof}
	The first part follows from $E(x) = \cost(x) = b^{\rot}p$ and the preceding Lemma. 
	Thus $x$ and $q$ converge to the set $F$ of equilibrium points, see \eqref{eq:defF}, 
	that are optimum solutions to $\eqref{ULP}$. 
	Since every optimum solution is an equilibrium point by 
	Theorem~\ref{optima are equilibria}, $x$ and $q$ converge to $\OPT$.
	For $e \not\in \Eopt$, $f_e = 0$ for every $f \in F\cap\OPT$.
	Since $x$ and $\abs{q}$ converge to $F\cap\OPT$,
	$x_e$ and $\abs{q_e}$ converge to zero for every $e \in \Eopt$.
\end{proof}

%% file: directed.tex
\section{Improved Convergence Results: Physarum-inspired dynamics}\label{directed discrete}\label{discretization}

In this section, we present in its full generality our main result on the 
Physarum-inspired dynamics \eqref{eq:DdirD}.

\subsection{Overview}

Inspired by the max-flow min-cut theorem, we consider the following primal-dual pair of linear programs:
the primal LP is given by 
$\max \set{t}{Af = t\cdot b;\ 0 \le f \le x}$
in variables $f \in \R^m$ and $t \in \R$, and its dual LP reads
$\min \set{ x^{\rot}z}{z \ge 0;\ z \ge A^{\rot}y;\ b^{\rot}y = 1}$
in variables $z \in \R^m$ and $y \in \R^n$. 
Since the dual feasible region does not contain a line and the minimum is bounded, 
the optimum is attained at a vertex, and in an optimum solution we have $z = \max\{0,A^{\rot}y\}$. 
Let $V$ be the set of vertices of the dual feasible region, and let 
$Y := \set{y}{(z,y) \in V}$ be the set of their projections on $y$-space.
Then, the dual optimum is given by
$\min\smSet{\max\{0,y^{\rot}A\}\cdot x}{y \in Y}$.
The set of \textit{strongly dominating} capacity vectors $x$ is defined as
\[
X := \set{x \in \R_{> 0}^m}{y^{\rot} A x > 0 \text{ for all } y \in Y}.\footnote{\label{fnExmSPP} 
	In the shortest path problem (recall that $b = e_1 - e_n$) 
	the set $\Y$ consists of all $y \in \sset{-1,+1}^n$ such that $y_1 = 1 = -y_n$, i.e., $y$ encodes a cut with $S = \set{i}{y_i = -1}$ and $\overline{S} = \set{i}{y_i = +1}$. 
	The condition $y^{\rot} A x > 0$ translates into 
	$\sum_{a\in E(S,\overline{S})}x_a-\sum_{a\in E(\overline{S},S) } x_a > 0$, i.e., every source-sink cut must have positive directed capacity. }
\]
Note that $\X$ contains the set of all scaled feasible solutions $\{x=tf\,:\,Af=b,\,f \ge 0,\,t>0\}$.

We next discuss the choice of step size. 
For $y \in Y$ and capacity vector $x$, let $\alpha(y,x) := y^{\rot} A x$. 
Further, let $\alpha(x) := \min \set{\alpha(y,x)}{y \in Y}$ and $\ahl:=\alpha(\xl)$. 
Then, for any $x\in X$ there is a feasible $f$ such that $0 \le f \le x/\alpha(x)$, 
see Lemma~\ref{lem:SDSgivesXCapF}. In particular, if $x$ is feasible then $\alpha(x)=1$,
since $\alpha(y,x) = 1$ for all $y \in Y$.
We partition the Physarum-inspired dynamics \eqref{eq:DdirD} into the following five regimes and 
define for each regime a fixed step size, see Subsection~\ref{subsec:SDCV}.

\begin{corollary}\label{cor:ahlConv}
	The Physarum-inspired dynamics \eqref{eq:DdirD} initialized with $\xz\in X$ and a step size $h$ satisfies:
	\begin{enumerate}
		\item If $\ahz = 1$, we work with $h \le \hprm $ and have $\ahl = 1$ for all $\ell$. 
		
		\item If $1/2 \leq \ahz < 1$, we work with $h \le \hprm /2$ and have $1-\delta\leq\ahl<1$ for 
		$\ell \ge h^{-1}\log(1/2\delta)$ and $\delta > 0$.
		
		\item If $1 < \ahz \leq 1/\hprm  $, we work with $h \le \hprm $ and have 
		$1<\ahl\leq1+\delta$ for $\ell \ge h^{-1}\cdot\log(1/\delta h_{0})$ and $\delta > 0$.
		
		\item If $0< \ahz < 1/2$, we work with $h \leq \ahz \hprm $ and have $1/2\leq\ahl <1$ for 
		$\ell \ge 1/h$.
		
		\item If $1/\hprm  < \ahz$, we work with $h \leq 1/4$ and have $1 < \ahl \leq 1/\hprm $
		for $\ell = \lfloor \log_{1/(1-h)}h_{0}(\alpha^{(0)}-1)/(1-h_{0}) \rfloor$.
	\end{enumerate}
	In each regime, we have $1 - \alpha^{(\ell + 1)} = (1 - h)(1 - \ahl)$.
\end{corollary}

We give now the full version of Theorem~\ref{thm_main} which applies 
for any strongly dominating starting point.

\begin{theorem}\label{thm_main_full}
	Suppose $A\in\Z^{n\times m}$ has full row rank $(n\leq m)$, $b\in\Z^{n}$, 
	$c\in\Z_{>0}^{m}$ and $\epsilon\in(0,1)$.
	Given $\xz \in X$ and its corresponding $\ah{0}$, 
	the Physarum-inspired dynamics \eqref{eq:DdirD} initialized with 
	$\xz$ runs in two regimes:
	\begin{enumerate}[\mbox{}\hspace{\parindent}(i)]
		\item The first regime is executed when $\ahz \not\in[1/2,1/\hprm]$ and it computes 
		a point $\x{t}\in X$ such that $\ah{t} \in[1/2,1/\hprm]$. In particular, if $\ahz < 1/2$ then 
		$h \leq (\Phi/\opt)\cdot(\ah{0}\hprm)^2$ and $t = 1/h$.
		Otherwise, if $\ahz > 1/\hprm$ then $h\leq\Phi/\opt$ and 
		$t = \lfloor \log_{1/(1-h)}[h_{0}(\alpha^{(0)}-1)/(1-h_{0})] \rfloor$.
		
		\item The second regime starts from a point $\x{t}\in X$ with $\ah{t} \in[1/2,1/\hprm]$, 
		it has a step size $h \leq (\Phi/\opt)\cdot h_{0}^{2}/2$ and outputs for any
		$k\geq4C_{1}/(h\Phi)\cdot\ln(C_{2}\PSI{0}/(\epsilon\cdot\min\{1,x_{\min}^{(0)}\}))$
		a vector $x^{(t+k)}\in X$ such that 
		$\dist(x^{(t+k)},\Xstar)<\epsilon/(\D \gamma_{A})$.
	\end{enumerate}
\end{theorem}

We stated the bounds on $h$ in terms of the unknown quantities $\Phi$ and $\opt$. However, $\Phi/\opt \ge 1/C_3$ by Lemma~\ref{lem:fsDg} and hence replacing $\Phi/\opt$ by $1/C_3$ yields constructive bounds for $h$.

\paragraph{Organization:} 

This section is devoted to proving Theorem~\ref{thm_main_full}, and it is organized as follows: 
Subsection~\ref{subsec:subsecUsefulLemmas} establishes core efficiency bounds that
extend \cite{SV-LP} and yield a \emph{scale-invariant} determinant dependence 
of the step size and are applicable to strongly dominating points.
Subsection~\ref{subsec:SDCV} gives the definition of strongly dominating points and
shows that the Physarum-inspired dynamics \eqref{eq:DdirD} initialized with such a point is
well defined.
Subsection~\ref{subsection:CloseNonNengKernelFreeVector} extends the analysis in
\cite{Physarum-Complexity-Bounds,SV-Flow,SV-LP} to positive linear programs, 
by generalizing the concept of non-negative flows to non-negative \emph{feasible kernel-free} vectors.
Subsection~\ref{subsec:epsCloseToOPT} shows that $\xl$ converges to $\Xstar$ for large enough $\ell$. 
Subsection~\ref{subsection:ProofThmMAIN} concludes the proof of Theorem~\ref{thm_main_full}.

\subsection{Useful Lemmas}\label{subsec:subsecUsefulLemmas}

Recall that $\Rl = \mathrm{diag}(c) \cdot (\Xl)^{-1}$ is a positive diagonal matrix and 
$\Ll  \overset{\mathrm{def}}{=}A(\Rl)^{-1}A^{\rot}$ is invertible. Let  $\pl$ be the unique solution
of $\Ll \pl=b$. We improve the dependence on $D_S$ in \cite[Lemma 5.2]{SV-LP} to $D$.

\begin{lemma}\cite[extension of Lemma 5.2]{SV-LP} \label{lem_ATLA_UB}
	Suppose $\xl>0$, $\Rl$
	is a positive diagonal matrix and $L=A(\Rl)^{-1}A^{\rot}$. Then for every
	$e\in[m]$, it holds that $\lVert A^{\rot}(\Ll)^{-1}A_{e}\rVert_{\infty}\leq
	\D  \cdot c_{e}/\xl_{e}$.
\end{lemma}
\begin{proof}
	The statement follows by combining the proof in \cite[Lemma 5.2]{SV-LP} with Lemma~\ref{lem:fsDg}.
\end{proof}

We show next that \cite[Corollary 5.3]{SV-Flow} holds for $x$-capacitated vectors,
which extends the class of feasible starting points, and further yields a bound in terms of $D$.

\begin{lemma}\cite[extension of Corollary 5.3]{SV-Flow} \label{lem_ATp_UB}
	Let $\pl$ be the unique solution of $\Ll\pl=b$ and assume $\xl$ is a positive vector
	with corresponding positive scalar $\ahl$ such that there is a vector $f$ satisfying 
	$Af=\ahl\cdot b$ and $0\leq f\leq \xl$.
	Then $\|A^{\rot}\pl\|_{\infty}\leq \D  \onenorm{c}/\ahl$.
\end{lemma}
\begin{proof}
	By assumption, $f$ satisfies $\ahl b=Af=\sum_{e}f_{e}A_{e}$ and
	$0\le f\le \xl$. This yields
	\begin{align*}
	\ahl\lVert A^{\rot}\pl\rVert_{\infty} & =\lVert A^{\rot}(\Ll)^{-1}\cdot\ahl b\rVert_{\infty} = 
	\lVert\sum_{e}f_{e}A^{\rot}(\Ll)^{-1}A_{e}\rVert_{\infty}\\
	& \leq\sum_{e}f_{e}\lVert A^{\rot}(\Ll)^{-1}A_{e}\rVert_{\infty}
	\overset{(\text{Lem. }\ref{lem_ATLA_UB})}{\leq}
	\D \sum_{e}f_{e}\frac{c_{e}}{\xl_{e}}\leq 
	\D \lVert c\rVert_{1}.\hfill\qedhere
	\end{align*}
\end{proof}
We note that applying Lemma~\ref{lem_ATLA_UB} and Lemma~\ref{lem_ATp_UB} into 
the analysis of \cite[Theorem 1.3]{SV-LP} yields an improved result that depends on
the scale-invariant determinant $D$. Moreover, we show in the next 
Subsection~\ref{subsec:SDCV} that the Physarum-inspired dynamics \eqref{eq:DdirD} 
can be initialized with any strongly dominating point.
\bigskip

We establish now an upper bound on $q$ that does not depend on $x$. We then use this upper bound on $q$ to establish a uniform upper bound on $x$. 

\begin{lemma}\label{lem:upper bound for q}
	For any $\xl>0$, $\lVert \ql\rVert_{\infty}\leq m\D^{2}\|b/\gamma_{A}\|_{1}$.
\end{lemma}

\begin{proof}
	Let $f$ be a basic feasible solution of $Af=b$.
	By definition, $\ql_{e}=(\xl_{e}/c_{e})A_{e}^{\rot}(\Ll)^{-1}b$ and thus
	\[
	\left|\ql_{e}\right|=\left|\frac{\xl_{e}}{c_{e}}\sum_{u}A_{e}^{\rot}(\Ll)^{-1}A_{u}f_{u}\right|
	\leq\frac{\xl_{e}}{c_{e}}\sum_{u}|f_{u}|\cdot\left|A_{e}^{\rot}(\Ll)^{-1}A_{u}\right|
	\leq\D \lVert f\rVert_{1},
	\]
	where the last inequality follows by
	\[
	\left|A_{e}^{\rot}(\Ll)^{-1}A_{u}\right|=\left|A_{u}^{\rot}(\Ll)^{-1}A_{e}\right|
	\leq\lVert A^{\rot}(\Ll)^{-1}A_{e}\rVert_{\infty}
	\overset{(\text{Lem. }\ref{lem_ATLA_UB})}{\leq}
	\D \cdot c_{e}/\xl_{e}.
	\]
	By Cramer's rule and Lemma~\ref{lem:fsDg}, we have 
	$|\ql_{e}|\leq\D \lVert f\rVert_{1}\leq m\D^{2}\|b/\gamma_{A}\|_{1}.$
\end{proof}

Let $k,t\in\mathbb{N}$. We denote by
\begin{equation}\label{eq:PhysDyn}
\overline{q}^{(t,k)}=
\sum_{i=t}^{t+k-1}\frac{h\left(1-h\right)^{t+k-1-i}}{1-(1-h)^{k}}q^{(i)}
\quad\text{and}\quad\overline{p}^{(t,k)}=\sum_{i=t}^{t+k-1}p^{(i)}.
\end{equation}
Straightforward checking shows that $A\overline{q}^{(t,k)}=b$. 
Further, for $C:=\mathrm{diag}(c)$, $t\geq0$ and $k\geq1$, we have
\[
\x{t}\prod_{i=t}^{t+k-1}[1+h(C^{-1}A^{\rot}p^{(i)}-1)]\,\,
=\,\,\x{t+k}\,\,=\,\,(1-h)^{k}\x{t}+[1-(1-h)^{k}]\overline{q}^{(t,k)}.
\]

We give next an upper bound on $\xk$ that is independent of $k$. 

\begin{lemma}\label{lem_UBxk}
	Let $\PSI{0}=\max\{m\D^{2}\|b/\gamma_{A}\|_{1},\|\xz \|_{\infty}\}$.
	Then $\lVert x^{(k)}\rVert_{\infty}\leq\PSI{0}$, $\forall k\in\mathbb{N}$.
\end{lemma}

\begin{proof}
	We prove the statement by induction. The base case $\|\xz \|_{\infty}\leq\PSI{0}$
	is clear. Suppose the statement holds for some $k>0$. Then, triangle
	inequality and Lemma \ref{lem:upper bound for q} yield
	\[
	\|x^{(k+1)}\|_{\infty}\le(1-h)\|x^{(k)}\|_{\infty}+h\|q^{(k)}\|_{\infty}\leq(1-h)\PSI{0}+h\PSI{0}\leq\PSI{0}.\hfill\qedhere
	\]
\end{proof}

We show now convergence to feasibility. 

\begin{lemma}\label{convergence to feasibility}
	Let $\rk = b - A \xk$. Then $\rkp =  (1 - h) \rk$ 
	and hence $\rk = (1 - h)^k (b - A \xz)$.
\end{lemma}
\begin{proof} 
	By definition $\xkp = (1 - h) \xk + h \qk$, and thus the statement follows by
	\[ 
	\rkp = b - A \xkp = b - (1 - h) A \xk - h b = (1 - h) \rk.
	\hfill\qedhere
	\]
\end{proof}

\subsection{Strongly Dominating Capacity Vectors}\label{subsec:SDCV}

For the shortest path problem, it is known that one can start from any capacity vector $x$ for which the directed capacity of every source-sink cut is positive, where the directed capacity of a cut is the total capacity of the edges crossing the cut in source-sink direction minus the total capacity of the edges crossing the cut in the sink-source direction. We generalize this result. 
We start with the max-flow like LP
\begin{equation}\label{eq:primalWDS}
\max \set{t}{Af = t\cdot b;\ 0 \le f \le x}
\end{equation}
in variables $f \in \R^m$ and $t \in \R$ and its dual
\begin{equation}\label{eq:dualWDS}
\min \set{ x^{\rot}z}{z \ge 0;\ z \ge A^{\rot}y;\ b^{\rot}y = 1} 
\end{equation}
in variables $z \in \R^m$ and $y \in \R^n$. The feasible region of the dual contains no line. Assume otherwise; say it contains $(z,y) = (\z{0},\y{0}) + \lambda (\z{1},\y{1})$ for all $\lambda \in \R$. 
Then, $z\ge 0$ implies $\z{1} = 0$ and further $z \ge A^{\rot}y$ implies 
$\z{0} \ge A^{\rot}\y{0} + \lambda A^{\rot}\y{1}$ and hence $A^{\rot}\y{1} = 0$. Since $A$ has full row rank, 
we have $\y{1} = 0$. The optimum of the dual is therefore attained at a vertex. 
In an optimum solution, we have $z = \max\{0,A^{\rot}y\}$. 
Let $V$ be the set of vertices of the feasible region of the dual~\eqref{eq:dualWDS}, and let 
\[    
Y := \set{y}{(z,y) \in V} 
\]
be the set of their projections on $y$-space. Then, the optimum of the dual~\eqref{eq:dualWDS} 
is given by 	
\begin{equation}\label{eq:dealWDSopt}
\min_{y \in Y} \left\{ \max\{0,y^{\rot}A\}\cdot x \right\}.
\end{equation}

The set of strongly dominating capacity vectors $x$ is defined by
\begin{equation}\label{eq:defSDS}
X := \set{x \in \R_{> 0}^m}{y^{\rot} A x > 0 \text{ for all } y \in Y}.
\end{equation}
We next show that for all $\xz \in X$ and sufficiently small step size, the sequence 
$\{\xk\}_{k\in\N}$ stays in $X$. Moreover, $y^{\rot} A \xk$ converges to 1 for every $y \in Y$.
We define by
\[
\alpha(y,x) := y^{\rot} A x\quad\text{ and }\quad\alpha(x) := \min \set{\alpha(y,x)}{y \in Y}.
\]
Let $\ahl:=\alpha(\xl)$. Then, $\xl\in X$ iff $\ahl>0$. We summarize the discussion
in the following Lemma.

\begin{lemma}\label{lem:SDSgivesXCapF}
	Suppose $\xl\in X$. Then, there is a vector $f$ such that $Af=\ahl\cdot b$ and 
	$0\leq f \leq \xl$.
\end{lemma}
\begin{proof}
	By the strong duality theorem applied on \eqref{eq:primalWDS} and \eqref{eq:dualWDS},
	it holds by \eqref{eq:dealWDSopt} that
	\[
		t = \min_{y \in Y} \left\{ \max\{0,y^{\rot}A\}\cdot \xl \right\}
		\geq \min_{y \in Y} y^{\rot}A \xl = \ahl.
	\]
	The statement follows by the definition of \eqref{eq:primalWDS}.
\end{proof}

We demonstrate now that $\ahl$ converges to $1$.

\begin{lemma}\label{lem:OneStep}
Assume $\xl \in X$. Then, for any $\hl \leq \min\{1/4,\ahl \hprm \}$ we have
$\xlplus \in X$ and $$1 - \ahlplus  = (1 - \hl)\cdot (1 - \ahl).$$
\end{lemma}
\begin{proof}
By applying Lemma~\ref{lem_ATp_UB} and Lemma~\ref{lem:SDSgivesXCapF} with $\xl\in X$, 
we have $\|A^{\rot}p^{(\ell)}\|_{\infty}\leq\D  \onenorm{c}/\ahl$ and hence for every index $e$
it holds 
$-\hl\cdot c_{e}^{-1}\xl_{e}A_{e}^{\rot}p^{(\ell)}\geq-(\hl\xl_{e})/(2\ahl h_{0})\geq-\xl_{e}/2$. 
Thus,
\[
	\xlplus_{e} =  (1- \hl)\xl_{e}+\hl\cdot[R_{e}^{(\ell)}]^{-1}A_{e}^{\rot}p^{(\ell)} 
	\ge \frac{3}{4} \xl_e - \frac{1}{2} \xl_e = \frac{1}{4} \xl_e > 0.
\]
Let $y \in \Y$ be arbitrary. Then $y^{\rot} b = 1$ and hence 
$y^{\rot} r^{(\ell)} = y^{\rot} (b - A\xl) = 1 - y^{\rot} A \xl = 1 - \alpha(y,\xl)$. 
The second claim now follows from Lemma~\ref{convergence to feasibility}.
\end{proof}

We note that the convergence speed crucially depends on the initial point $\x{0}\in X$, 
and in particular to its corresponding value $\ah{0}$. Further, this dependence naturally 
partitions the Physarum-inspired dynamics \eqref{eq:DdirD} into the five regimes given in 
Corollary~\ref{cor:ahlConv}.

\subsection{$x^{(k)}$ is Close to a Non-Negative Kernel-Free Vector}
\label{subsection:CloseNonNengKernelFreeVector}

In this subsection, we generalize \cite[Lemma 5.4]{SV-Flow} to positive linear programs. 
We achieve this in two steps. First, we generalize a result by Ito et al.~\cite[Lemma 2]{Ito-Convergence-Physarum}
to positive linear programs and then we substitute the notion of a non-negative cycle-free flow 
with a non-negative feasible kernel-free vector.

Throughout this and the consecutive subsection, we denote by
$\rhoA :=\max\left\{D\gamma_{A},\,n\D^{2}\infnorm{A}\right\}$.
\begin{lemma}\label{lem_RoundingFlows}
	Suppose a matrix $A\in\mathbb{Z}^{n\times m}$
	has full row rank and vector $b\in\mathbb{Z}^{n}$. Let $g$ be a
	feasible solution to $Ag=b$ and $S\subseteq[n]$ be a subset of row indices of $A$
	such that $\sum_{i\in S}|g_{i}|<1/\rhoA $. Then, there
	is a feasible solution $f$ such that $g_{i}\cdot f_{i}\ge0$ for
	all $i\in[n]$, $f_{i}=0$ for all $i\in S$ and $\lVert f-g\rVert_{\infty}<1/(\D \gamma_{A})$.
\end{lemma}
\begin{proof}
	W.l.o.g. we can assume that $g\ge0$ as we could change the signs
	of the columns of $A$ accordingly. Let $\mathbf{1}_{S}$ be the indicator
	vector of $S$. We consider the linear program 
	\[
	\min\{\mathbf{1}_{S}^{\rot}x\,:\,Ax=b,\ x\ge0\}
	\]
	and let $opt$ be its optimum value. 
	Notice that $0\le opt\le\mathbf{1}_{S}^{\rot}g<1/\rhoA$.
	Since the feasible region does not contain a line and the minimum
	is bounded, the optimum is attained at a basic feasible solution,
	say $f$.
	Suppose that there is an index $i\in S$ with $f_{i}>0$.
	By Lemma~\ref{lem:fsDg}, we have $f_{i}\ge1/(\D \gamma_{A})$. 
	This is a contradiction to the optimality of $f$ and hence $f_{i}=0$
	for all $i\in S$.
	
	Among the feasible solutions $f$ such that $f_i g_i \ge 0$ for all $i$ and $f_i = 0$ for all 
	$i \in S$, we choose the one that minimizes $\infnorm{f-g}$. For simplicity, 
	we also denote it by $f$.
	Note that $f$ satisfies $\supp(f)\subseteq\overline{S}$, where $\overline{S}=[m]\backslash S$.
	Further, since $f_{S}=0$ and 
	\[
	A_{S}g_{S}+A_{\overline{S}}g_{\overline{S}}=Ag=b=Af=A_{S}f_{S}+A_{\overline{S}}f_{\overline{S}}=A_{\overline{S}}f_{\overline{S}}
	\]
	we have $A_{\overline{S}}\left(f_{\overline{S}}-g_{\overline{S}}\right)=A_{S}g_{S}$.
	Let $A_{B}$ be a linearly independent column subset of $A_{\overline{S}}$
	with maximal cardinality, i.e. the column subset $A_{N}$, where $N = \overline{S} \setminus B$, is linearly
	dependent on $A_{B}$.
	Hence, there is an invertible square submatrix $A_{B}^{\prime}\in\mathbb{Z}^{|B|\times|B|}$
	of $A_{B}$ and a vector $v=(v_{B},0_{N})$ such that 
	\[
	\left(\begin{array}{c}
	A_{B}^{\prime}\\
	A_{B}^{\prime\prime}
	\end{array}\right)v_{B}
	=A_{B}v_{B}
	=A_{S}g_{S}.
	\]
	Let $r=\left(A_{S}g_{S}\right)_{B}$.
	Since $A_{B}^{\prime}$ is invertible, there is a unique vector $v_{B}$
	such that $A_{B}^{\prime}v_{B}=r$.
	Observe that
	\[
	|r_i|=\left|\sum_{j\in S}A_{i,j}g_{j}\right|\leq\inftynorm{A}\sum_{j\in S}|g_{j}| < \frac{\inftynorm{A}}{nD^{2}\inftynorm{A}}=\frac{1}{nD^{2}}.
	\]
	By Cramer's rule $v_{B}(e)$ is quotient of two determinants. 
	The denominator is $\det(A_{B}^{\prime})$ and hence at least one in absolute value. 
	For the numerator, the $e$-th column is replaced by $r$. 
	Expansion according to this column shows that the absolute value of the numerator is bounded by 
	\[ 
		\frac{D}{\gamma_{A}}\sum_{i\in B}\left| r_i \right|<\frac{D}{\gamma_{A}}\cdot\frac{|B|}{nD^{2}}\le \frac{1}{D\gamma_{A}}.
	\]
	Therefore, $\lVert f-g\rVert_{\infty}\leq1/(D\gamma_{A})$ and the statement follows.
\end{proof}

\begin{lemma}\label{lem_FeasibleKernelFree}
	Let $q\in\R^m$, $p\in\R^n$ and 
	$N=\{e\in [m]: q_{e}\leq0\text{ or } p^{\rot}A_{e}\le 0\}$,
	where $Aq=b$ and $p=L^{-1}b$.
	Suppose $\sum_{e\in N}|q_{e}|<1/\rhoA$.
	Then there is a non-negative feasible kernel-free vector $f$ 
	such that $\supp(f)\subseteq E\backslash N$
	and $\lVert f-q\rVert_{\infty}<1/(\D \gamma_{A})$.
\end{lemma}
\begin{proof}
	We apply Lemma~\ref{lem_RoundingFlows} to $q$ with $S=N$. Then, there is
	a non-negative feasible vector $f$ such that $\supp(f)\subseteq E\backslash N$
	and $\lVert f-q\rVert_{\infty}<1/(\D \gamma_{A})$.
	By Lemma~\ref{sign-compatible representation},
	$f$ can be expressed as a sum of a convex combination of basic feasible solutions
	plus a vector $w$ in the kernel of $A$. Moreover, all vectors in
	this representation are sign compatible with $f$, and in particular
	$w$ is non-negative too.
	
	Suppose for contradiction that $w\neq 0$. 
	By definition, $0=p^{\rot}Aw=\sum_{e\in[m]}p^{\rot}A_{e}w_e$
	and since $w\geq0$ and $w\neq0$, it follows that
	there is an index $e\in[m]$ satisfying
	$w_e>0$ and $p^{\rot}A_{e}\leq0$.	
	Since $f$ and $w$ are sign compatible, $w_e>0$ implies $f_e>0$.
	On the other hand, as $p^{\rot}A_{e}\leq0$ we have $e\in N$ and thus $f_{e}=0$.
	This is a contradiction, hence $w=0$.
\end{proof}

Using Corollary~\ref{cor:ahlConv}, for any point $\x{0}\in X$ there is a point
$\x{t}\in X$ such that $\ah{t}\in[1/2,1/\hprm]$. Thus, we can assume that 
$\ah{0}\in[1/2,1/\hprm]$ and work with $h \le \hprm /2$, where $\hprm  = c_{\min}/(2\D \onenorm{c})$.
We generalize next \cite[Lemma 5.4]{SV-Flow}.

\begin{lemma}\label{lem_ApproxNonNegKernelFree}
    Suppose $\x{t}\in X$ such that $\ah{t}\in[1/2,1/\hprm]$,
	$h\leq \hprm /2$ and $\epsilon\in(0,1)$.
	Then, for any $k\geq h^{-1}\ln(8m\rhoA \PSI{0}/\epsilon)$
	there is a non-negative feasible kernel-free vector $f$ such that
	$\lVert x^{(t+k)}-f\rVert_{\infty}<\epsilon/(\D \gamma_{A})$.
\end{lemma}
\begin{proof}
	Let $\beta^{(k)}\overset{\text{def}}{=}1-(1-h)^{k}$. By \eqref{eq:PhysDyn},
	vector $\overline{q}^{(t,k)}$ satisfies $A\overline{q}^{(t,k)}=b$
	and thus Lemma \ref{lem_UBxk} yields
	\begin{equation}\label{eq:infNormXQ}
	\lVert x^{(t+k)}-\beta^{(k)}\overline{q}^{(t,k)}\rVert_{\infty}=(1-h)^{k}\cdot
	\lVert \x{t}\rVert_{\infty}\leq \exp\{-hk\}\cdot\PSI{0}\le \epsilon/(8m\rhoA ).
	\end{equation}
	Using Corollary~\ref{cor:ahlConv}, we have $x^{(t+k)}\in X$ such that 
	$\alpha^{(t+k)}\in(1/2,1/\hprm)$ for every $k\in\N_+$.
	Let $F_{k}=Q_{k}\cup P_{k}$, where
	$Q_{k}=\{e\in[m] : \overline{q}_{e}^{(t,k)}\leq0\}$
	and $P_{k}=\{e\in[m] : A_{e}^{\rot}\overline{p}^{(t,k)}\leq0\}$.
	Then, for every $e\in Q_{k}$ it holds
	\begin{equation}\label{eq:Qk}
	|\overline{q}_{e}^{(t,k)}|
	\leq[\beta^{(k)}]^{-1}\cdot|x_{e}^{(t+k)}-
	\beta^{(k)}\overline{q}_{e}^{(t,k)}|
	\leq \epsilon/(7m\rhoA ).
	\end{equation}
	By Lemma~\ref{lem_UBxk}, $\lVert x^{(\cdot)} \rVert \leq \PSI{0}$. Moreover, by (\ref{eq:PhysDyn}) for every $e\in P_{k}$ we have
	\begin{eqnarray*}
		x_{e}^{(t+k)} & = & x_{e}^{(t)}\prod_{i=t}^{k+t-1}
		\left[1+h\left(c_{e}^{-1}A_{e}^{\rot}p^{(i)}-1\right)\right]\\
		& \leq & x_{e}^{(t)}\cdot\exp\left\{ -hk + (h/c_{e})\cdot 
		A_{e}^{\rot}\overline{p}^{(t,k)}\right\} \\
		& \leq & \exp\left\{ -hk\right\} \cdot \PSI{0}\\
		& \leq & \epsilon/(8m\rhoA ),
	\end{eqnarray*}
	and by combining the triangle inequality with (\ref{eq:infNormXQ}), it
	follows for every $e\in P_{k}$ that
	\begin{eqnarray}
	|\overline{q}_{e}^{(t,k)}| & \leq & [\beta^{(k)}]^{-1}\cdot\left[|x_{e}^{(t+k)}-\beta^{(k)}\overline{q}_{e}^{(t,k)}|+|x_{e}^{(t+k)}|\right]\nonumber \\
	& \leq & [\beta^{(k)}]^{-1}\cdot \epsilon/(4m\rhoA )\nonumber \\
	& \leq & \epsilon/(3m\rhoA ). \label{eq:Pk}
	\end{eqnarray}
	Therefore, (\ref{eq:Qk}) and (\ref{eq:Pk}) yields that
	\begin{equation}\label{eq:sumQ1D}
	\sum_{e\in F_{k}}|\overline{q}_{e}^{(t,k)}|
	\leq m \cdot \epsilon/(3m\rhoA )
	\leq \epsilon/(3\rhoA ).
	\end{equation}
	
	By Lemma \ref{lem_FeasibleKernelFree} applied with $\overline{q}_{e}^{(t,k)}$
	and $N=F_{k}$, it follows by (\ref{eq:sumQ1D}) that there is a non-negative feasible
	kernel-free vector $f$ such that $\supp(f)\subseteq E\backslash N$
	and 
	\[
	\lVert f-\overline{q}^{(t,k)}\rVert_{\infty}<\epsilon/(3\D \gamma_{A}).
	\]
	By Lemma \ref{lem:upper bound for q}, we have $\lVert\overline{q}^{(t,k)}\rVert_{\infty}\leq m\D^{2}\lVert b/\gamma_{A}\rVert_{1}$
	and since $\PSI{0}\geq m\D^{2}\|b/\gamma_{A}\|_{1}$, it follows that
	\begin{align*}
		\lVert x^{(t+k)}-f\rVert_{\infty} & =\lVert x^{(t+k)}-\beta^{(k)}\overline{q}^{(t+k)}+\beta^{(k)}\overline{q}^{(t+k)}-f\rVert_{\infty}\\
		& \leq\lVert x^{(t+k)}-\beta^{(k)}\overline{q}^{(t+k)}\rVert_{\infty}+\lVert f-\overline{q}^{(t+k)}\rVert_{\infty}+\left(1-h\right)^{k}\lVert\overline{q}^{(t+k)}\rVert_{\infty}\\
		& \leq
		\frac{\epsilon}{8m\rhoA }+\frac{\epsilon}{3\D \gamma_{A}}+
		\frac{\epsilon\cdot m\D^{2}\lVert b/\gamma_{A}\rVert_{1}}{8m\rhoA \cdot\PSI{0}}
		\leq\frac{\epsilon}{\D \gamma_{A}}.\qedhere
	\end{align*}
\end{proof}

\subsection{$x^{(k)}$ is $\epsilon$-Close to an Optimal Solution}\label{subsec:epsCloseToOPT}

Recall that $\mathcal{N}$ denotes the set of non-optimal basic feasible solutions 
of \eqref{OLP} and $\Phi=\min_{g\in\mathcal{N}}c^{\rot}g-\opt$. For completeness,
we prove next a well known inequality~\cite[Lemma 8.6]{PapadimitriouSteiglitz82}
that lower bounds the value of $\Phi$.

\begin{lemma}\label{lem_lowBounds}
	Suppose $A\in\R^{n\times m}$ has full row rank, $b\in\R^n$ and $c\in\R^m$ 
	are integral. Then, $\Phi\geq1/(\D\gamma_{A})^2$.
\end{lemma}
\begin{proof}
	Let $g=(g_{B},0)$ be an arbitrary basic feasible solution with basis matrix $A_B$,
	where $g_{B}(e)\neq0$ and $|\mathrm{supp}(g_{B})|=n$. 
	We write $M_{-i,-j}$ to denote the matrix $M$ with deleted $i$-th row 
	and $j$-th column. Let $Q_e$ be the matrix formed by replacing the $e$-th column 
	of $A_{B}$ by the column vector $b$.
	Then, by Cramer's rule, we have
	\[
	\left|g_{B}(e)\right| = \left|\frac{\det(Q_e)}{\det(A_{B})}\right| = \frac{1}{\gamma_A}\left|\sum_{k=1}^{n}\frac{\left(-1\right)^{j+k}\cdot b_{k}\cdot\det\left(\gamma_{A}^{-1}[A_{B}]_{-k,-j}\right)}
	{\det\left(\gamma_{A}^{-1}A_{B}\right)}\right| \geq\frac{1}{D\gamma_{A}}.
	\]
	Note that all components of vector $g_{B}$ have denominator with equal value, 
	i.e. $\det(A_{B})$. Consider an arbitrary non-optimal basic feasible solution $g$ 
	and an optimal basic feasible solution $\fstr$. Then, $g_{e}=G_{e}/G$ and $\fstr_{e}=F_{e}/F$ 
	are rationals such that $G_{e},G,F_{e},F\leq D\gamma_{A}$ for every $e$.
	Further, let $r_{e}=c_{e}\left(G_{e}F - F_{e}F\right)\in\mathbb{Z}$ for every $e\in[m]$, 
	and observe that
	\[
		c^{\rot}\left(g-\fstr\right) = \sum_{e}c_{e}\left(g_{e}-\fstr_{e}\right) = 
		\frac{1}{GF}\sum_{e}r_{e}\geq1/(D\gamma_{A})^{2},
	\]
	where the last inequality follows by $c^{\rot}(g-\fstr)>0$ implies $\sum_{e}r_{e}\geq1$.
\end{proof}

\begin{lemma}\label{lem_optCrit}
    Let $f$ be a non-negative feasible kernel-free vector and
    $\epsilon\in\left(0,1\right)$ a parameter. Suppose for every non-optimal
    basic feasible solution $g$, there exists an index $e\in[m]$
    such that $g_{e}>0$ and $f_{e}<\epsilon/(2m\D^{3}\gamma_{A}\onenorm{b})$.
    Then, $\lVert f-\fstr \rVert_{\infty}<\epsilon/(\D \gamma_{A})$ for
    some optimal $\fstr $.
\end{lemma}
\begin{proof}
	Let $C=2\D^{2}\onenorm{b}$. 
	Since $f$ is kernel-free, by Lemma~\ref{sign-compatible representation} it can
	be expressed as a convex combination of sign-compatible basic feasible solutions $f=\sum_{i=1}^{\ell}\alpha_{i}f^{(i)}+\sum_{i=\ell+1}^{m}\alpha_{i}f^{(i)}$,
	where $f^{(1)},\ldots,f^{(\ell)}$ denote the optimal solutions.
	By Lemma~\ref{lem:fsDg}, $f_e^{(i)} > 0$ implies $f_e^{(i)}\geq1/(\D \gamma_{A})$.
	By the hypothesis, for every non-optimal $f^{(i)}$, i.e. $i\geq\ell+1$,
	there exists an index $e(i)\in[m]$ such that 
	\[
	1/(\D\gamma_{A})\leq f_{e(i)}^{(i)}\quad\text{and}\quad f_{e(i)} <\epsilon/(m\D\gamma_{A}\cdot C).
	\]
	Therefore, we have
	\[
		\alpha_{i}/(\D\gamma_{A}) 
		\leq \alpha_{i}f_{e(i)}^{(i)}\leq\sum_{j=1}^{m}\alpha_{j}f_{e(i)}^{(j)}
		 = f_{e(i)}<\epsilon/(m\D\gamma_{A}\cdot C),
	\]
	and hence $\sum_{i=\ell+1}^{m}\alpha_{i}\leq\epsilon/C$.	
	Further, by Lemma~\ref{lem:fsDg}, for every $j$ we have 
	\[
		\lVert f^{(j)} \rVert_\infty \le \D\lVert b/\gamma_{A}\rVert_{1}=C/(2\D\gamma_{A}).
	\]
	Let $\beta\geq0$ be an arbitrary vector satisfying $\sum_{i=1}^{\ell}\beta_{i}=\sum_{i=\ell+1}^{m}\alpha_{i}$. 
	Let $\nu_{i}=\alpha_{i}+\beta_{i}$ for every $i\in[\ell]$ and let  
	$\fstr =\sum_{i=1}^{\ell}\nu_{i}f^{(i)}$. 
	Then, $\fstr $ is an optimal solution and we have 
\begin{align*}
\lVert \fstr  - f \rVert_\infty &= \left\Vert \sum_{i=1}^{\ell}\beta_{i}f^{(i)} - \sum_{i=\ell+1}^{m}\alpha_{i}\cdot f^{(i)} \right\Vert _{\infty}  \\
&\le \max_{i\in[1:m]}\left\Vert f^{(i)}\right\Vert _{\infty}\cdot\left(\sum_{i=1}^{\ell}\beta_{i}+\sum_{i=\ell+1}^{m}\alpha_{i}\right) \\
&\le\frac{2\epsilon}{C}\cdot\frac{C}{2\D\gamma_{A}}=\frac{\epsilon}{\D\gamma_{A}}.
\hfill\qedhere
\end{align*}
\end{proof}

In the following lemma, we extend the analysis in~\cite[Lemma 5.6]{SV-Flow} from the 
transshipment problem to positive linear programs.
Our result crucially relies on an argument that uses the parameter
$\Phi = \min_{g \in \mathcal{N}} c^{\rot} g - \opt$. 
It is here, where our analysis incurs the linear step size dependence on $\Phi/\opt$ 
and the quadratic dependence on $\opt/\Phi$ for the number of steps.

An important technical detail is that the first regime incurs an extra 
$(\Phi/\opt)$-factor dependence.
At first glance, this might seem unnecessary due to Corollary~\ref{cor:ahlConv}, 
however a careful analysis shows its necessity (see \eqref{eq:InductiveArgument}
for the inductive argument).
Further, we note that the undirected Physarum dynamics~\eqref{eq:DUndirD} 
satisfies $\x{t}_{\min}\geq(1-h)^t\cdot\x{0}_{\min}$, whereas
the directed Physarum-inspired dynamics \eqref{eq:DdirD} might yield a
value $\x{t}_{\min}$ which decreases with faster than exponential rate. 
As our analysis incurs a logarithmic dependence on $1/\x{0}_{\min}$, it is prohibitive 
to decouple the two regimes and give bounds in terms of $\log(1/\x{t}_{\min})$, which 
would be necessary as $\x{t}$ is the initial point of the second regime.

\begin{lemma}\label{lem_PotArg}
	Let $g$ be an arbitrary non-optimal basic feasible solution.
	Given $\xz\in X$ and its corresponding $\ah{0}$, the Physarum-inspired dynamics \eqref{eq:DdirD} 
	initialized with $\xz$ runs in two regimes:
	\begin{compactenum}[\mbox{}\hspace{\parindent}(i)]
	\item The first regime is executed when $\ahz \not\in[1/2,1/\hprm]$ and computes 
	a point $\x{t}\in X$ such that $\ah{t}\in[1/2,1/\hprm]$. 
	In particular, if $\ahz<1/2$ then $h \leq (\Phi/\opt)\cdot(\ah{0}\hprm)^2$ and $t = 1/h$.
	Otherwise, if $\ahz > 1/\hprm$ then $h\leq\Phi/\opt$ and 
	$t=\lfloor \log_{1/(1-h)}[h_{0}(\alpha^{(0)}-1)/(1-h_{0})] \rfloor$.
	
	\item The second regime starts from a point $\x{t}\in X$ such that 
	$\ah{t} \in[1/2,1/\hprm]$, it has step size $h \leq (\opt/\Phi)\cdot h_{0}^{2}/2$ and 
	for any $k\geq 4\cdot c^{\rot}g/(h\Phi)\cdot\ln(\PSI{0}/\epsilon x_{\min}^{(0)})$,
	guarantees the existence of an index $e\in[m]$ such that $g_{e}>0$
	and $x_{e}^{(t+k)}<\epsilon$.
	\end{compactenum}
\end{lemma} 

\begin{proof}
	Similar to the work of~\cite{Physarum-Complexity-Bounds,SV-Flow}, we use a potential function that takes
	as input a basic feasible solution $g$ and a step number $\ell$, and is defined by
	\[
	\mathcal{B}_{g}^{(\ell)}:=\sum_{e\in[m]}g_{e}c_{e}\ln x_{e}^{(\ell)}.
	\]
	Since $x_{e}^{(\ell+1)}=x_{e}^{(\ell)}(1+\hl[c_{e}^{-1}\cdot A_{e}^{\rot}p^{(\ell)}-1])$, we have 
	\begin{align}
	\mathcal{B}_{g}^{(\ell+1)}-\mathcal{B}_{g}^{(\ell)} & =\sum_{e}g_{e}c_{e}\ln\frac{x_{e}^{(\ell+1)}}{x_{e}^{(\ell)}}=\sum_{e}g_{e}c_{e}\ln\left(1+\hl\left[\frac{A_{e}^{\rot}p^{(\ell)}}{c_{e}}-1\right]\right)\nonumber \\
	& \leq\hl\sum_{e}g_{e}c_{e}\left[\frac{A_{e}^{\rot}p^{(\ell)}}{c_{e}}-1\right]=\hl\left[-c^{\rot}g+[p^{(\ell)}]^{\rot}Ag\right]=\hl\left[-c^{\rot}g+b^{\rot}p^{(\ell)}\right].\label{eq:g}
	\end{align}
	Let $\fstr$ be an optimal solution to \eqref{OLP}. In order to lower
	bound $\mathcal{B}_{\fstr}^{(\ell+1)}-\mathcal{B}_{\fstr}^{(\ell)}$,
	we use the inequality $\ln(1+x)\geq x-x^{2}$, for all $x\in[-\frac{1}{2},\frac{1}{2}]$.
	Then, we have 
	\begin{eqnarray}
	\mathcal{B}_{\fstr}^{(\ell+1)}-\mathcal{B}_{\fstr}^{(\ell)} & = & \sum_{e}f_{e}^{\star}c_{e}\ln\frac{x_{e}^{(\ell+1)}}{x_{e}^{(\ell)}}=\sum_{e}f_{e}^{\star}c_{e}\ln\left(1+\hl\left[\frac{A_{e}^{\rot}p^{(\ell)}}{c_{e}}-1\right]\right)\nonumber \\
	& \geq & \sum_{e}f_{e}^{\star}c_{e}\left(\hl\left[\frac{A_{e}^{\rot}p^{(\ell)}}{c_{e}}-1\right]-
	[\hl]^{2}\left[\frac{A_{e}^{\rot}p^{(\ell)}}{c_{e}}-1\right]^{2}\right) \nonumber \\
	& \geq & \hl\left(b^{\rot}p^{(\ell)}-\opt-\hl\cdot(1/2\ahl\hprm)^{2}\cdot\opt\right),\label{eq:one}
	\end{eqnarray}
	where the last inequality follows by combining 
	\[
	\sum_{e}f_{e}^{\star}c_{e}[(c_{e}^{-1}A_{e}^{\rot}p^{(\ell)})-1]=[p^{(\ell)}]^{\rot}A\fstr-\opt=b^{\rot}p^{(\ell)}-\opt,
	\]
	$\|A^{\rot}p^{(\ell)}\|_{\infty}\leq\D\onenorm{c}/\ahl$ (by Lemma~\ref{lem_ATp_UB}
	and Lemma~\ref{lem:SDSgivesXCapF} applied with $\xl\in X$), $h_0 = c_{\min}/(4\D \onenorm{c})$ 
	and
	\[
		\hl\sum_{e}f_{e}^{\star}c_{e}\cdot(c_{e}^{-1}A_{e}^{\rot}p^{(\ell)}-1)^{2}
		\leq\hl(2\D\onenorm{c}/\ahl c_{\min})^{2}\opt=\hl(1/2\ahl\hprm)^{2}\opt.
	\]
	Further, by combining (\ref{eq:g}), (\ref{eq:one}), $c^{\rot}g-\opt\ge\Phi$
	for every non-optimal basic feasible solution $g$ and provided that 
	the inequality $\hl(1/2\ahl\hprm)^{2}\opt\leq\Phi/2$ holds, we obtain 
	\begin{equation}\label{eq:IterativePotArgB}
	\mathcal{B}_{\fstr}^{(\ell+1)}-\mathcal{B}_{\fstr}^{(\ell)}\geq
	\hl\left(b^{\rot}p^{(\ell)}-c^{\rot}g\right)+	\hl\left(c^{\rot}g-\opt-\frac{\Phi}{2}\right)
	\geq\mathcal{B}_{g}^{(\ell+1)}-\mathcal{B}_{g}^{(\ell)}+\frac{\hl\Phi}{2}.
	\end{equation}

	Using Corollary~\ref{cor:ahlConv}, we partition the Physarum-inspired dynamics \eqref{eq:DdirD} 
	execution into three regimes, based on $\ah{0}$. For every $i\in\{1,2,3\}$, 
	we show next that the $i$-th regime has a fixed step size $\hl=h_i$ such that
	$\hl(1/2\ahl\hprm)^{2}\opt\leq\Phi/2$, for every step $\ell$ in this regime.

	By Lemma~\ref{lem:OneStep}, for every $i\in\{1,2,3\}$ it holds for every step $\ell$ 
	in the $i$-th regime that
	\begin{equation}\label{eq:recAHL}
		\ahl=1-(1-h_i)^{\ell}\cdot(1-\ah{0}).
	\end{equation}
	
	\textbf{Case 1:} Suppose $\ah{0}>1/\hprm$.
	Notice that $\hl=\Phi/\opt$ suffices, since $1/(2\ahl\hprm)<1/2$
	for every $\ahl>1/\hprm$.
	Further, by applying \eqref{eq:recAHL} with $\ah{t}:=1/\hprm$, 
	we have $t=\lfloor \log_{1/(1-\hl)}[h_{0}(\alpha^{(0)}-1)/(1-h_{0})] \rfloor\leq(\opt/\Phi)\cdot\log(\ah{0}\hprm)$.
	Note that by~\eqref{eq:recAHL} the sequence $\{\ahl\}_{\ell\leq t}$ is decreasing, and 
	by Corollary~\ref{cor:ahlConv} we have $1<\ah{t}\leq1/\hprm $.	
	
	\textbf{Case 2:} Suppose $\ah{0}\in(0,1/2)$. By \eqref{eq:recAHL}
	the sequence $\{\ahl\}_{\ell\in\N}$ is increasing and by Corollary
	\ref{cor:ahlConv} the regime is terminated once $\ahl\in[1/2,1)$.
	Observe that $\hl=(\Phi/\opt)\cdot(\ah{0}\hprm)^{2}$ suffices,
	since $\ah{0}\leq\ahl$. Then, by \eqref{eq:recAHL} applied with $\ah{t}:=1/2$, 
	this regime has at most $t=(\opt/\Phi)\cdot(1/\ah{0}\hprm)^{2}$ steps.
	
	\textbf{Case 3:} Suppose $\ah{0}\in[1/2,1/\hprm]$. By \eqref{eq:recAHL}
	the sequence $\{\ahl\}_{\ell\in\N}$ converges to $1$ (decreases
	if $\ah{0}\in(1,1/\hprm]$ and increases when $\ah{0}\in[1/2,1)$.
	Notice that $\hl=(\Phi/\opt)\cdot\hprm^{2}/2$ suffices, since
	$1/2\leq\ahl\leq1/\hprm$ for every $\ell\in\N$.
	We note that the number of steps in this regime is to be determined soon.
	
	Hence, we conclude that inequality \eqref{eq:IterativePotArgB} holds.	
	Further, using Case 1 and Case 2 there is an integer $t\in\N$ 
	such that $\ah{t}\in[1/2,1/\hprm]$. Let $k\in\N$ be the number of steps in Case 3,
	and let $h:=(\Phi/\opt)\cdot\hprm^{2}/2$.
	Then, for every $\ell\in\{t,\ldots,t+k-1\}$ it holds that $\hl=h$ and thus
	\begin{equation}\label{eq:InductiveArgument}
	\mathcal{B}_{\fstr}^{(t+k)}-\mathcal{B}_{\fstr}^{(0)}\geq
	\mathcal{B}_{g}^{(t+k)}-\mathcal{B}_{g}^{(0)}+\sum_{\ell=0}^{t+k-1}\frac{\hl\Phi}{2}
	\geq\mathcal{B}_{g}^{(t+k)}-\mathcal{B}_{g}^{(0)}+k\cdot\frac{h\Phi}{2}.
	\end{equation}
	By Lemma $\ref{lem_UBxk}$, $\mathcal{B}_{g}^{(\ell)}\leq c^{\rot}g\cdot\ln\PSI{0}$
	for every basic feasible solution $g$ and every $\ell\in\N$, and thus
	\begin{eqnarray*}
		\mathcal{B}_{g}^{(t+k)}&\leq&-k\cdot\frac{h\Phi}{2}+\mathcal{B}_{g}^{(0)}+
		\mathcal{B}_{\fstr}^{(t+k)}-\mathcal{B}_{\fstr}^{(0)}\\
		&\leq&-k\cdot\frac{h\Phi}{2}+c^{\rot}g\cdot\ln\PSI{0}+
		\opt\cdot\ln\PSI{0}-\opt\cdot\ln x_{\min}^{(0)}\\
		&\leq&-k\cdot\frac{h\Phi}{2}+2c^{\rot}g\cdot\ln\frac{\PSI{0}}{x_{\min}^{(0)}}.
	\end{eqnarray*}
	Suppose for the sake of a contradiction that for every $e\in[m]$ with
	$g_{e}>0$ it holds $x_{e}^{(t+k)}>\epsilon$. 
	Then, $\mathcal{B}_{g}^{(t+k)}>c^{\rot}g\cdot\ln\epsilon$ yields
	$k < 4\cdot c^{\rot}g/(h\Phi)\cdot\ln(\PSI{0}/(\epsilon x_{\min}^{(0)}))$,
	a contradiction to the choice of $k$. 
\end{proof}

\subsection{Proof of Theorem~\ref{thm_main_full}}\label{subsection:ProofThmMAIN}
By Corollary~\ref{cor:ahlConv} and Lemma~\ref{lem_PotArg}, if $\xz\in X$ 
such that $\ahz>1/\hprm $, we work with $h\leq\Phi/\opt$ and after 
$t=\lfloor \log_{1/(1-h)}[h_{0}(\alpha^{(0)}-1)/(1-h_{0})] \rfloor
\leq(\opt/\Phi)\cdot\log(\ah{0}\hprm)$
steps, we obtain $\x{t}\in X$ such that $\ah{t}\in(1,1/\hprm ]$.
Otherwise, if $\ahz\in(0,1/2)$ we work with $h\leq(\Phi/\opt)\cdot(\ah{0}\hprm)^{2}$ 
and after $t=1/h$ steps, we obtain $\x{t}\in X$ such that $\ah{t}\in[1/2,1)$.
Hence, we can assume that $\ah{t} \in [1/2,1/\hprm ]$ and set $h\leq(\Phi/\opt)\cdot\hprm^{2}$. 
Then, the Lemmas in Subsection~\ref{subsection:CloseNonNengKernelFreeVector} and 
\ref{subsec:epsCloseToOPT} are applicable.

Let $E_1:=\D\|b/\gamma_{A}\|_{1}\onenorm{c}$, 
$E_2:=8m\rhoA\Psi^{(0)}$, 
$E_3:=2m\D^{3}\gamma_{A}\lVert b\rVert_{1}$ and
$E_4 := 8m\D^{2}\lVert b\rVert_{1}$.
Consider an arbitrary non-optimal basic feasible solution $g$.

By Lemma~\ref{lem:fsDg}, we have $c^{\rot}g\leq E_1$ and thus both 
Lemma~\ref{lem_ApproxNonNegKernelFree} and Lemma~\ref{lem_PotArg} 
are applicable with $h$, $\epsStr :=\epsilon/E_4$ and any
$k\geq k_{0}:=4E_{1}/(h\Phi)\cdot\ln[(E_{2}/\min\{1,x_{\min}^{(0)}\})\cdot(D\gamma_{A}/\epsStr)]$.
Hence, by Lemma \ref{lem_PotArg}, the Physarum-inspired dynamics \eqref{eq:DdirD} 
guarantees the existence of an index $e\in[m]$ such that $g_{e}>0$ and
$x_{e}^{(t+k)}<\epsStr/(D\gamma_{A})$.
Moreover, by Lemma \ref{lem_ApproxNonNegKernelFree} there is a non-negative 
feasible kernel-free vector $f$ such that 
$\lVert x^{(t+k)}-f\rVert_{\infty}< \epsStr /(\D \gamma_{A})$. Thus, for the index $e$
it follows that
$g_{e}>0$ and 
$f_{e}<2\epsStr/\D\gamma_{A}=(\epsilon/2)\cdot(4/E_{4}\D\gamma_{A})=\epsilon/(2E_{3})$.
Then Lemma \ref{lem_optCrit}, yields 
$\lVert f-\fstr \rVert_{\infty}<\epsilon/(2\D \gamma_{A})$ and by
triangle inequality we have
$\lVert x^{(k)}-\fstr \rVert_{\infty}<\epsilon/(\D \gamma_{A})$.

By construction, $\rhoA =\max\{D\gamma_{A},\,n\D^{2}\infnorm{A}\}\leq 
n\D^{2}\gamma_{A}\infnorm{A}$. Let 
$E_2^{\prime}=8mn\D^{2}\gamma_{A}\infnorm{A}\Psi^{(0)}$ and
$E_5=E_{2}^{\prime}E_{4}\cdot\D\gamma_{A}
=8^{2}m^{2}n\D^{5}\gamma_{A}^{2}\infnorm{A}\lVert b\rVert_{1}$.
Further, let $C_1=E_1$ and $C_2=E_5$. Then, the statement follows for any 
$k\geq k_{1}:=4C_{1}/(h\Phi)\cdot\ln(C_{2}\Psi^{(0)}/(\epsilon\cdot\min\{1,x_{\min}^{(0)}\}))$.
\qed

%% file: Preconditioning.tex
\subsection{Preconditioning}\label{sec:Preconditioning}

In this subsection, we generalize the preconditioning technique developed in \cite{Physarum-Complexity-Bounds,SV-Flow} for flow problems,
to the setting of positive linear programs. 

\begin{theorem}\label{thm_main_simpl}
	Given an integral LP $(A,\,b,\,c>0)$, a positive $\xz\in\R^{m}$ and a
	parameter $\epsilon\in(0,1)$. Let $([A\,|\,b],\,b,\,(c,\,c'))$ be an extended LP with 
	$c'= 2C_1$ and 
	$\z{0} := 1+D_S\inftynorm{x}\onenorm{A}\onenorm{b}$.\footnote{
		We denote by $\onenorm{A}:=\sum_{i,j}|A_{i,j}|$, i.e. we interpret 
		matrix $A$ as a vector and apply to it the standard $\ell_1$ norm.}
	Then, $(\xz; \z{0})$ is a strongly dominating starting point of the extended problem 
	such that $y^{\rot} [A\,\vert\, b] (\xz,\,\z{0}) \ge 1$, for all $y \in Y$.
	In particular, the Physarum-inspired dynamics \eqref{eq:DdirD} initialized with
	$(\xz,\,\z{0})$ and a step size $h\leq h_{0}^{2}/C_3$, outputs for any 
	$k\geq4C_{1}\cdot(\D \gamma_{A})^{2}/h\cdot
	\ln(C_{2}\UPSILON{0}/(\epsilon\cdot\min\{1,x_{\min}^{(0)}\}))$
	a vector $(\x{k},\z{k})>0$ such that 
	$\dist(\x{k},\Xstar)<\epsilon/(\D \gamma_{A})$ and
	$\z{k}<\epsilon/(\D \gamma_{A})$,
	where $\UPSILON{0}:=\max\{\PSI{0},\z{0}\}$.
\end{theorem}

Theorem~\ref{thm_main_simpl} subsumes~\cite[Theorem 1.2]{SV-Flow} for flow problems
by giving a tighter asymptotic convergence rate, 
since for the transshipment problem $A$ is a totally unimodular matrix 
and satisfies $\D=D_S=1$, $\gamma_{A}=1$, $\inftynorm{A}=1$ and $\Phi=1$.
We note that the scalar $\z{0}$ depends on the scaled determinant $D_S$, 
see Theorem~\ref{thm:Prev}.

\subsubsection{Proof of Theorem~\ref{thm_main_simpl}}

In the extended problem, we concatenate to matrix $A$ a column equal to $b$ such that
the resulting constraint matrix becomes $[A\,\vert\, b]$.
Let $c'$ be the cost and let $x'$ be the initial capacity of the newly
inserted constraint column. We will determine $c'$ and $x'$ in the course of the discussion. 
Consider the dual of the max-flow like LP for the extended problem. 
It has an additional variable $z'$ and reads
\[      
\min \set{x^{\rot}z + x'z'}{z \ge 0;\ z' \ge 0;\ z \ge A^{\rot}y;\ 
	z' \ge b^{\rot}y;\ b^{\rot}y = 1}.
\]
In any optimal solution, $z' = b^{\rot}y = 1$ and hence the dual is equivalent to 
\begin{equation}\label{eq:dualExt}
	\min \set{x^{\rot}z + x' }{z \ge 0;\ z\ge A^{\rot}y;\ b^{\rot} y = 1}.
\end{equation}
The strongly dominating set of the extended problem is therefore equal to 
\begin{equation}\label{defSDSExt}
	X = \set{\ex \in \R_{> 0}^{m+1}}{y^{\rot} [A\,\vert\, b] \ex > 0 \text{ for all } y \in Y}.
\end{equation}
The defining condition translates into $x' >  - y^{\rot} A x$ for all $y \in Y$. We summarize the discussion in the following Lemma.

\begin{lemma} 
	Given a positive $x\in\R^m$, let $\rho := \onenorm{b}D_S$ and
	$x' := 1+\rho\onenorm{A}\inftynorm{x}$, where
	$\onenorm{A}:=\sum_{i,j}|A_{i,j}|$ and 
	$\D_S := \max\{|\det(A')|\,:\, A'\text{ is a square sub-matrix of }A\}$.
	Then, $(x;\ x')$ is a strongly dominating 
	starting point of the extended problem such that
	$y^{\rot} [A\,\vert\, b] (x;\ x') = y^{\rot} A x + x' \ge 1$, for all $y \in Y$.
\end{lemma}
\begin{proof}
	We show first that $\max_{y\in Y}\inftynorm{y}\leq\rho$ implies the statement. Let $y \in Y$ be arbitrary. 
	Since $|y^{\rot} A x|\leq\onenorm{A}\inftynorm{x}\inftynorm{y}$, 
	we have $\max_{y\in Y}|y^{\rot} A x|\leq \rho\onenorm{A}\inftynorm{x}=x'-1$ and hence
	$y^{\rot} [A\,\vert\, b] (x;\ x') \ge 1$.
	
	It remains to show that $\max_{y\in Y}\inftynorm{y}\leq\rho$.
	The constraint polyhedron of the dual~\eqref{eq:dualExt} is given in matrix notation as
	\[
	\PExt:=\left\{ \left(\begin{array}{c}
	z\\
	y
	\end{array}\right)\in\R^{m+n}\,:\,\left[\begin{array}{cc}
	I_{m\times m} & -A^{\rot}\\
	\mathbf{0}_{m}^{\rot} & b^{\rot}\\
	I_{m\times m} & \mathbf{0}_{m\times n}
	\end{array}\right]\left(\begin{array}{c}
	z\\
	y
	\end{array}\right)\begin{array}{c}
	\geq\\
	=\\
	\geq
	\end{array}\left(\begin{array}{c}
	\mathbf{0}_{m}\\
	1\\
	\mathbf{0}_{m}
	\end{array}\right)\right\}.
	\]
	Let us denote the resulting constraint matrix and vector by $M\in\R^{2m+1 \times m+n}$
	and $d\in\R^{2m+1}$, respectively.
	
	Note that if $b=\vZeros$ then the primal LP~\eqref{eq:primalWDS} 
	is either unbounded or infeasible. Hence, we consider the non-trivial case 
	when $b\neq\vZeros$. Observe that the polyhedron $\PExt$ is not empty, 
	since for any $y$ such that $b^{\rot}y=1$ there is $z=\max\{0,A^{\rot}y\}$ satisfying 
	$(z;\,y)\in\PExt$.
	Further, $\PExt$ does not contain a line (see Subsection~\ref{subsec:SDCV}) 
	and thus $\PExt$ has at least one extreme point $p^{\prime}\in \PExt$. 
	As the dual LP~\eqref{eq:dualWDS} has a bounded value (the target function is lower 
	bounded by $0$) and an extreme point exists ($p^{\prime}\in \PExt$), the optimum 
	is attained at an extreme point $p\in \PExt$.
	Moreover, as every extreme point is a basic feasible solution
	and matrix $M$ has linearly independent columns ($A$ has full row rank),
	it follows that $p$ has $m+n$ tight linearly independent constraints.
	
	Let $M_{B(p)}\in\R^{m+n\times m+n}$ be the basis submatrix of $M$ satisfying $M_{B(p)}p=d_{B(p)}$.
	Since $A,b$ are integral and $M_{B(p)}$ is invertible, 
	using Laplace expansion we have
	$1\leq|\det(M_{B(p)})|\leq\onenorm{b}\D_S=\rho$.
	Let $Q_i$ denotes the matrix formed by replacing the $i$-th column 
	of $M_{B(p)}$ by the column vector $d_{B(p)}$. Then, by Cramer's rule, it follows that	
	$|y_i|=|\det(Q_i)/\det(M_{B(p)})|
	\leq|\det(M_{B(p)})|\leq\rho$,	for all $i\in[n]$.
\end{proof}

It remains to fix the cost of the new column. Using Lemma~\ref{lem:fsDg}, 
$\opt\leq c^{\rot}\x{k}\leq C_1$ for every $k\in\N$, 
and thus we set $c' := 2C_1$.

%% file: LowerBound.tex
\subsection{A Simple Lower Bound}\label{SimpleLowerBound}

Building upon~\cite[Lemma B.1]{SV-Flow}, we give a lower bound on 
the number of steps required for computing an $\epsilon$-approximation to 
the optimum shortest path.
In particular, we show that for the Physarum-inspired dynamics \eqref{eq:DdirD} to compute 
a point $\x{k}$ such that $\dist(\x{k},\Xstar)<\epsilon$, 
the required number of steps $k$ has to grow linearly in $\opt/(h\Phi)$ and $\ln(1/\epsilon)$.

\begin{theorem}\label{thm_one} 
	Let $(A,b,c)$ be a positive LP instance such that $A=[1\,\,1]$, $b=1$ and 
	$c=(\opt,\,\,\opt+\Phi)^{\rot}$, where $\opt>0$ and $\Phi>0$.
	Then, for any $\epsilon\in(0,1)$ the discrete directed Physarum-inspired dynamics \eqref{eq:DdirD}
	initialized with $\x{0}=(1/2,1/2)$ and any step size $h\in(0,1/2]$, 
	requires at least $k=(1/2h) \cdot \max\{\opt/\Phi,1\}\cdot\ln(2/\epsilon)$ steps	
	to guarantee $\x{k}_{1}\geq1-\epsilon$, $\x{k}_{2}\leq\epsilon$. 
	Moreover, if $\epsilon\leq\Phi/(2\opt)$ then
	$c^{\rot} \xk \ge (1 + \epsilon)\opt$ as long as 
	$k \le (1/2h) \cdot \max\{\opt/\Phi,1\}\cdot\ln(2 \Phi/(\epsilon\cdot\opt))$.
\end{theorem}
\begin{proof} 
Let $c_1=\opt$ and $c_2=\gamma \opt$, where $\gamma = 1 + \Phi/\opt$. 
We first derive closed-form expressions for $\xk_1$, $\xk_2$, and $\xk_1 + \xk_2$. 
Let $\sk = \gamma \xk_1 + \xk_2$. 
For any $k\in\N$, we have $\qk_1 + \qk_2 = 1$ and 
$\qk_1/\qk_2 = (\xk_1/c_1)/(\xk_2/c_2) = \gamma \xk_1/\xk_2$. 
Therefore, $\qk_1 = \gamma \xk_1/\sk$ and $\qk_2 = \xk_2/\sk$, and hence 
\begin{equation}\label{eq:xkp1}
\xk_{1}=(1+h(-1+\gamma/\skm))\xkm_{1}\quad\text{and}\quad\xk_{2}=(1+h(-1+1/\skm))\xkm_{2}.
\end{equation}
Further,  $\x{k}_{1}+\x{k}_{2}=(1-h)(\xkm_{1} +\xkm_2)+h$, and thus by induction 
$\x{k}_{1}+\x{k}_{2}= 1$ for all $k\in\N$. 

Therefore, $\sk \le \gamma$ for all $k\in\N$ and hence
$\xk_1 \geq \xkm_1$, i.e. the sequence $\{\xk_1\}_{k \in\N}$ is increasing and 
the sequence $\{\xk_2\}_{k \in\N}$ is decreasing. 
Moreover, since $h(-1+1/\skm) \geq h(1-\gamma)/\gamma = -h\Phi/(\opt+\Phi)$ 
and using the inequality $1-z\geq e^{-2z}$ for every $z\in[0,1/2]$, 
it follows by \eqref{eq:xkp1} and induction on $k$ that
\[
\xk_{2} \ge \left(1-\frac{h\Phi}{\opt+\Phi}\right)^{k}\x{0}_{2}
\ge \frac{1}{2}\exp\left\{ -k\cdot\frac{2h\Phi}{\opt+\Phi}\right\}.
\]
Thus, $\xk_2 \ge \epsilon$ whenever $k\leq(1/2h)\cdot(\opt/\Phi+1)\cdot\ln(2/\epsilon)$.
This proves the first claim. 

For the second claim, observe that 
$c^{\rot} \xk = \opt\cdot \xk_1 + \gamma \opt\cdot \xk_2 = \opt\cdot (1 + (\gamma - 1)\xk_2)$. 
This is greater than $(1 + \epsilon)\opt$ iff $\xk_2 \ge \epsilon\cdot\opt/\Phi$. 
Thus, $c^{\rot} \xk \ge (1 + \epsilon)\opt$ as long as 
$k\leq(1/2h)\cdot(\opt/\Phi+1)\cdot\ln(2/(\epsilon\cdot\opt/\Phi))$.
\end{proof}